%% file: lipics-v2021-sample-article.tex
\documentclass[a4paper,UKenglish,cleveref, autoref, thm-restate]{lipics-v2021}

\pdfoutput=1 
\hideLIPIcs  

\title{Popularity on the Roommate Diversity Problem}

\author{Steven {Ge}}{Department of Mathematical and Computing Science, Tokyo Institute of Technology, Japan \and \url{https://ge-steven.github.io/} }{ge.s.aa@m.titech.ac.jp}{https://orcid.org/0000-0001-5073-748X}{}

\author{Toshiya {Itoh}}{Department of Mathematical and Computing Science, Tokyo Institute of Technology, Japan}{titoh@c.titech.ac.jp}{https://orcid.org/0000-0002-1149-7046}{}

\authorrunning{S. Ge and T. Itoh} 

\Copyright{Steven Ge and Toshiya Itoh} 



\category{} 

\relatedversion{} 

\nolinenumbers 

\EventEditors{John Q. Open and Joan R. Access}
\EventNoEds{2}
\EventLongTitle{42nd Conference on Very Important Topics (CVIT 2016)}
\EventShortTitle{CVIT 2016}
\EventAcronym{CVIT}
\EventYear{2016}
\EventDate{December 24--27, 2016}
\EventLocation{Little Whinging, United Kingdom}
\EventLogo{}
\SeriesVolume{42}
\ArticleNo{23}

\begin{document}

\maketitle

\begin{abstract}
A recently introduced restricted variant of the multidimensional stable roommate problem is the roommate diversity problem\cite{10.5555/3491440.3491454}: each agent belongs to one of two types (e.g., red and blue), and the agents' preferences over the coalitions solely depend on the fraction of agents of their own type among their roommates. A roommate diversity game is represented by the quadruple $G = (R,B,s, (\succsim_a)_{a \in R\cup B})$, where $R$ and $B$ represent the red and blue agents respectively, $s \in \mathbb{N}^+$ denotes the room size, and preference relation $\succsim_a$ for $a \in R\cup B$ is a weak order over $\{\frac{j}{s} | j \in [0,s]\}$. 

There are various notions of stability that defines an optimal partitioning of agents. The notion of popularity has received a lot of attention recently. A partitioning of agents is popular if there does not exist another partitioning in which more agents are better off than worse off. Computing a popular partition in a stable roommate game can be done in polynomial time\cite{IRVING1985577}. When we allow ties the stable roommate problem becomes NP-complete\cite{RONN1990285}. Determining the existence of a popular solution in the multidimensional stable roommate problem is also NP-hard\cite{doi:10.1137/0404023}.

We show that in the roommate diversity problem with $s=2$ fixed, the problem becomes tractable. Particularly, a popular partitioning of agents is guaranteed to exist and can be computed in polynomial time. Additionally a mixed popular partitioning of agents is always guaranteed to exist in any roommate diversity game. By contrast, when there are no restrictions on the coalition size of a roommate diversity game, a popular partitioning may fail to exist and the problem becomes intractable. Our intractability results are summarized as follows:
\begin{itemize}
	\item Determining the existence of a popular partitioning is co-NP-hard, even if the agents' preferences are trichotomous.
	\item Determining the existence of a strictly popular partitioning is co-NP-hard, even if the agents' preferences are dichotomous.
	\item Computing a mixed popular partitioning of agents in polynomial time is impossible unless P=NP, even if the agents' preferences are dichotomous.
\end{itemize}
\end{abstract}

\newpage
\input{sections/introduction}
\input{sections/preliminaries}
\input{sections/roomsizetwo}
\input{sections/strictpopularity}
\input{sections/mixedpopularity}
\input{sections/popularity}
\input{sections/conclusion}

\bibliography{lipics-v2021-sample-article}

\end{document}

%% file: sections/introduction.tex
\section{Introduction}
The formation of stable coalitions in multi-agents systems is a computational problem that has several variations with different conditions on the coalition size. A recently introduced restricted variant of the multidimensional stable roommate problem is the roommate diversity problem\cite{10.5555/3491440.3491454}: each agents belongs to one of two types (e.g., red and blue), and the agents' preferences over the coalitions solely depend on the fraction of agents of their own type among their roommates. This model captures important aspects of several real-world coalition formation scenarios, such as flat-sharing, seating arrangements at events, and splitting students into teams for group projects. In the latter scenario, as local student may have difficulty communicating with international students and vice versa, the preference of a student may solely rely on the fraction of international group members.

There are various notions of stability that defines an optimal partitioning of agents. The notion of popularity has received a lot of attention recently. A partitioning of agents is popular if there does not exist another partitioning in which more agents are better off than worse off. Computing a popular partition in a stable roommate game can be done in polynomial time\cite{IRVING1985577}. When we allow ties the stable roommate problem becomes NP-complete\cite{RONN1990285}. Determining the existence of a popular solution in the multidimensional stable roommate problem is also NP-hard\cite{doi:10.1137/0404023}.

The roommate diversity problem has been studied with the stability notions envy-freeness, Pareto optimality, exchange stability, and core stability. For the majority of these stability notions computing a stable partitioning of a roommate diversity game is NP-hard. On the other hand the problem is FPT when using the room size as parameter\cite{10.5555/3491440.3491454}.

\subsection{Our results}
We show that in the roommate diversity problem with the room size fixed to 2, the problem becomes tractable. Particularly, a popular partitioning of agents is guaranteed to exist and can be computed in polynomial time. Additionally a mixed popular partitioning of agents is always guaranteed to exist in any roommate diversity game. By contrast, when there are no restrictions on the coalition size of a roommate diversity game, a popular partitioning may fail to exist and the problem becomes intractable. Our intractability results are summarized as follows:
\begin{itemize}
	\item Determining the existence of a popular partitioning is co-NP-hard, even if the agents' preferences are trichotomous.
	\item Determining the existence of a strictly popular partitioning is co-NP-hard, even if the agents' preferences are dichotomous.
	\item Computing a mixed popular partitioning of agents in polynomial time is impossible unless P=NP, even if the agents' preferences are dichotomous.
\end{itemize}

%% file: sections/preliminaries.tex
\section{Preliminaries}\label{prelim}
In this work, we slightly extend the notation and definitions related to the roommate diversity problem used in the work of Boehmer and Elkind\cite{10.5555/3491440.3491454}. Additionally, we also use the notation and definitions related to popularity used in the work of Brandt and Bullinger\cite{10.1613/jair.1.13470}. In order to ensure that this paper is self-contained, we shall include the existing definitions in this section. For our hardness results, we construct reductions from the Exact Cover by 3-Sets problem. We abbreviate Exact Cover by 3-Sets with X3C. These reductions can be performed in polynomial time and the size of the resulting roommate diversity game is polynomial in the size of the X3C instance.

\subsection{Roommate Diversity Problem}
For $s \in \mathbb{N}, t \in \mathbb{N}^+$ let us define the integer set $[t] = \{1, \dots, t\}$ and $[s,t] = \{s, \dots, t\}$.

\begin{definition}
	A roommate diversity game is a quadruple $G = (R,B,s, (\succsim_a)_{a \in R\cup B})$ with room size $s$ and agent set $N = R \cup B$ where $|N| = k \cdot s$ for some $k \in \mathbb{N}$. The preference relation $\succsim_a$ of each agent $a \in N$ is a complete and transitive weak order over the set $D = \{\frac{j}{s} | j \in [0,s]\}$.
\end{definition}

We call the agents in $R$ red agents and the agents in $B$ blue agents. A $s$-sized subset of $N$ is called a coalition or room and the number of rooms is $k = \frac{|N|}{s}$. For an agent $a \in N$ let $\mathcal{N}_a = \{S \subseteq N | |S| = s, a \in S\}$ denote every possible room that contains $a$.

An outcome $\pi$ of $G$ is a partition of the agents $N$ into $k$ rooms of size $s$, i.e., $\pi = \{C_1, \dots, C_k\}$ such that $|C_i| = s$ for each $i \in [k]$, $\bigcup\limits_{i = 1}^{k}C_i = N$, and for $1 \leq i < j \leq k$ we have $C_i \cap C_j = \emptyset$. Let $\pi(a)$ denote the room in $\pi$ that contains agent $a \in N$, i.e., for $a\in N$, we have $\pi(a) = C_i$ where $i \in [k]$ such that $a \in C_i$. For a room $C \subseteq N$, let $\theta(C)$ denote the fraction of red agents in $C$, i.e., $\theta(C) = \frac{|C \cap R|}{|C|}$. We say that $C$ has/is of fraction $\theta(C)$ or the fraction of $C$ is $\theta(C)$.

For an agent $a \in N$, the preference relation $\succsim_a$ are the preferences of $a$ over the fraction of red agents in its room. For example, $\frac{2}{s} \succsim_a \frac{3}{s}$ means that agent $a$ likes being in a room with 2 red agents at least as much as being in a room with $3$ red agents. Note that a red agent cannot be in a room with fraction $\frac{0}{s}$ and a blue agent cannot be in a room with fraction $\frac{s}{s}$. Thus discarding these `impossible' fractions in their preference relation does not impact our results. Given two rooms $S, T \in \mathcal{N}_a$, we overload the notation by writing $S \succ_a T$, and say that agent $a$ strictly prefers $S$ over $T$ if $\theta(S) \succsim_a \theta(T)$ and $\theta(T) \not\succsim_a \theta(S)$. Additionally, we write $S \succsim_a T$ and say that agent $a$  weakly prefers $S$ over $T$ if $\theta(S) \succsim_a \theta(T)$. If agent $a$ weakly prefers $S$ over $T$ and $T$ over $S$, we write $S \sim_a T$ and say that $a$ is indifferent between $S$ and $T$. We similarly overload the notation for outcomes $\pi, \pi'$ by writing $\pi \succ_a \pi'$ iff $\pi(a) \succ_a \pi'(a)$, $\pi \succsim_a \pi'$ iff $\pi(a) \succsim_a \pi'(a)$,  and $\pi \sim_a \pi'$ iff $\pi(a) \sim_a \pi'(a)$.


The preference relation of agent $a$ is said to be dichotomous if there exists a partition of $D$ into two sets $D^+_a$ and $D^-_a$ such that for all $d^+\in D^+_a$, $d^- \in D^-_a$ it holds that $d^+ \succ_a d^-$, for all $d^+_1, d^+_2\in D^+_a$ it holds that $d^+_1 \sim_a d^+_2$, and for all $d^-_1, d^-_2\in D^-_a$ it holds that $d^-_1 \sim_a d^-_2$. We say that agent $a$ approves of the fractions in $D^+_a$ and disapproves of the fractions in $D^-_a$.

The preference relation of agent $a$ is said to be trichotomous if there exists a partition of $D$ into three sets $D^+_a$, $D^n_a$, and $D^-_a$ such that for all $d^+\in D^+_a$, $d^n \in D^n_a$, and $d^- \in D^-_a$ it holds that $d^+ \succ_a d^n$ and $d^n \succ_a d^-$. Additionally, for all $d^+_1, d^+_2\in D^+_a$ it holds that $d^+_1 \sim_a d^+_2$, for all $d^n_1, d^n_2\in D^n_a$ it holds that $d^n_1 \sim_a d^n_2$, and for all $d^-_1, d^-_2\in D^-_a$ it holds that $d^-_1 \sim_a d^-_2$. We say that agent $a$ approves of the fractions in $D^+_a$, is neutral about the fraction in $D^n_a$, and disapproves of the fractions in $D^-_a$.

For an outcome $\pi$ let $D_\pi^+$ denote the agents that approve of the room it is assigned in $\pi$, i.e., $D_\pi^+ = \{a \in R \cup B | \theta(\pi(a)) \in D_a^+\}$. We define $D_\pi^n$ and $D_\pi^-$ analogously.

\subsection{Popularity}
For outcomes $\pi, \pi'$ let $N(\pi, \pi')$ be the set of agents who prefer $\pi$ over $\pi'$, i.e., $N(\pi,\pi') = \{a \in N | \pi \succ_a \pi'\}$. For any room $S \subseteq N$, where $|S|=s$, and outcomes $\pi, \pi'$ let $\phi_S(\pi,\pi') = |N(\pi, \pi') \cap S| - |N(\pi',\pi) \cap S|$. We call $\phi_S(\pi,\pi')$ the popularity margin on $S$ with respect to $\pi$ and $\pi'$. We define the popularity margin of $\pi$ and $\pi'$ as $\phi(\pi, \pi') = \phi_N(\pi, \pi')$. 

An outcome $\pi$ is more popular than outcome $\pi'$ if $\phi(\pi,\pi') > 0$. An outcome $\pi$ is popular if for any outcome $\pi'$ we have $\phi(\pi,\pi') \geq 0$, i.e., no outcome is more popular than $\pi$. An outcome $\pi$ is called strongly popular if for any other outcome $\pi'\neq \pi$ we have $\phi(\pi,\pi')>0$, i.e., $\pi$ is more popular than any other outcome. Note that there can be at most one strongly popular outcome.

We define a mixed outcome $p = \{(\pi_1,p_1), \dots, (\pi_t,p_t)\}$ to be a set of pairs, where for each $i \in [t]$, $\pi_i$ is an outcome of a roommate diversity game and $(p_1, \dots, p_t)$ is a probability distribution. For mixed outcome $p = \{(\pi_1,p_1), \dots, (\pi_t,p_t)\}$ and $q = \{(\sigma_1,q_1), \dots, (\sigma_u,q_u)\}$, we define the popularity margin of $p$ and $q$ to be their expected popularity margin, i.e.,
$$\phi(p,q) = \sum\limits_{i = 1}^{t}\sum\limits_{j = 1}^{u}p_iq_j\phi(\pi_i,\sigma_j).$$
A mixed outcome $p$ is popular if for any mixed outcome $q$ we have $\phi(p,q)\geq 0$.

\subsection{Exact Cover by 3-Sets Problem}
For our reductions, we use the exact cover by 3-set problem, which is known to be NP-complete\cite{10.5555/574848}. Let $X = \{1, \dots, m\} = [m]$, where $m \in \mathbb{N}^+$, and let $C = \{A_1, \dots, A_q\}$ be a collection of 3-element subsets of $X$, i.e., for each $i \in [q]$ we have $A_i \subseteq X$ and $|A_i| = 3$. An instance of the X3C problem is a tuple $(X,C)$ and asks: does there exist a subset $C' \subseteq C$ such that every element of $X$ occurs in exactly one member of $C'$? 
We call such a $C'$ a solution of $(X,C)$.

We require the following definition for our reduction. For $i \in X$, let $J^i = \{j^i_1, \dots, j^i_{m_i}\}$ be the set of indices of the sets in $C$ that contain $i$, i.e. $j \in J^i \iff i \in A_j$ (or equivalently $J^i = \{j \in [q] | i \in A_j\}$).

%% file: sections/roomsizetwo.tex
\section{Room Size Two}
For a roommate diversity game $G = (R,B, 2, (\succsim_a)_{a \in R\cup B})$ with room size 2, a room $S \subseteq R\cup B$, can have exactly 3 possible fraction, i.e., $\theta(S) \in \{\frac{0}{2}, \frac{1}{2}, \frac{2}{2}\}$. Let us call a room with fraction $\frac{0}{2}$ or $\frac{2}{2}$ a pure blue or pure red room respectively. A room with fraction $\frac{1}{2}$ we call a mixed room. 

As mentioned in \cref{prelim}, we can discard the `impossible' fractions from the preference relation of the agents. Thus the only relevant fractions for a red agent are $\frac{1}{2}$ and $\frac{2}{2}$ and the only relevant fractions for a blue agent are $\frac{0}{2}$ and $\frac{1}{2}$. Let us call an agent that is in a room with one of their most preferred fraction happy. Otherwise we call the agent sad. That is, an agent $a \in R \cup B$ is happy in outcome $\pi$ if for each $f\in \{\frac{0}{2}, \frac{1}{2}, \frac{2}{2}\}$ we have $\theta(\pi(a)) \succsim_a f$. An agent $a \in R \cup B$ is sad in outcome $\pi$ if there exists $f\in \{\frac{0}{2}, \frac{1}{2}, \frac{2}{2}\}$ such that $\theta(\pi(a)) \prec_a f$.

A red agent $r$ can only have one of 3 possible preference relation $\frac{1}{2}\succ_r \frac{2}{2}$, $\frac{1}{2}\prec_r \frac{2}{2}$, or $\frac{1}{2}\sim_r \frac{2}{2}$. We call a red agent $r$ with preference relation $\frac{1}{2}\succ_r \frac{2}{2}$, $\frac{1}{2}\prec_r \frac{2}{2}$, or $\frac{1}{2}\sim_r \frac{2}{2}$ a mixed, pure, or indifferent red agent respectively. We define mixed, pure, and indifferent blue agents using fraction $\frac{0}{2}$ and $\frac{1}{2}$ analogously.

Let us define the set of pure red agents $R^p = \{r \in R | \frac{2}{2}\succ_r \frac{1}{2}\}$, the set of mixed red agents $R^{m} = \{r \in R | \frac{2}{2}\prec_r \frac{1}{2}\}$, and the set of indifferent red agents $R^i = \{r \in R | \frac{1}{2}\sim_r \frac{2}{2}\}$. We define the set of pure blue agents $B^p$, the set of mixed blue agents $B^{m}$, and the set of indifferent blue agents $B^i$ analogously. Note that $R = R^p \cup R^{m} \cup R^i$ and $B = B^p \cup B^{m} \cup B^i$.

We show that a popular outcome is guaranteed to exist in roommate diversity game $G$ and can be computed in polynomial time by reducing $G$ to the maximum weight perfect matching problem.

Let us define an undirected weighted clique graph $G_m = (R \cup B, E)$ where for each pair of distinct agents $a, b \in R\cup B$, we use
$w(a, b)$ to denote the weight of an edge $(a, b) \in E$, where
$$w(a,b) = \begin{cases}
2 & (a, b) \text{ is a pair with 2 happy agents};\\
1 & (a, b) \text{ is a pair with exacly 1 happy agent}; \\
0 & (a, b) \text{ is a pair with exacly 0 happy agents}.
\end{cases}$$
Since the room size is 2, we have that $|R \cup B|$ is even. For the weighted graph $G_m$, let $w(M)$ be the weight of a perfect matching $M \subseteq E$. Given a perfect matching $M$ of $G_m$, we define a
point $p_M(a)$ of an agent $a \in R \cup B$ as follows: For each $(a, b) \in M$,
\begin{enumerate}
	\item if $w(a, b) = 2$, then $p_M(a) = p_M(b) = 1$;
	\item \begin{enumerate}
		\item if $w(a, b) = 1$ where $a$ is happy and $b$ is sad, then $p_M(a) = 1$ and $p_M(b) = 0$;
		\item if $w(a, b) = 1$ where $a$ is sad and $b$ is happy, then $p_M(a) = 0$ and $p_M(b) = 1$;
	\end{enumerate}
	\item if $w(a, b) = 0$, then $p_M(a) = p_M(b) = 0$.
\end{enumerate}
We have that for any perfect matching $M$ of $G_m$,
\begin{align}
	w(M) = \sum\limits_{(a,b)\in M}w(a,b) = \sum\limits_{a \in R\cup B}p_M(a).\label{weigheq}
\end{align}

\begin{lemma}\label{roomsize2lem}
	For the weighted graph $G_m = (R \cup B, E)$ as defined above, the maximum weight perfect matching $M_*$ of G is popular.
\end{lemma}
\begin{proof}
	Fix an arbitrary perfect matching $M$ of $G_m$. It is immediate that $w(M_*) \geq w(M)$. For the maximum weight perfect matching $M_*$ of $G_m$, we define $H_{M_*}$ and $S_{M_*}$ as
	\begin{align*}
		& H_{M_*} = \{a \in R \cup B | p_{M_*}(a) = 1\}; &
		S_{M_*} = \{a \in R \cup B | p_{M_*}(a) = 0\},
	\end{align*}
	and for the perfect matching $M$ of $G_m$, we also define $H_M$ and $S_M$ by
	\begin{align*}
		& H_{M} = \{a \in R \cup B | p_{M}(a) = 1\}; &
		S_{M} = \{a \in R \cup B | p_{M}(a) = 0\}.
	\end{align*}
	From \cref{weigheq} and the fact that $w(M_*) \geq w(M)$, it follows that
	\begin{align}
		|H_{M_*}| = w(M_*) \geq w(M) = |H_M|. \label{hweight}
	\end{align}
	We have that $\phi(M_*, M) = |H_{M_*}|-|H_{M_*} \cap H_M|$ and $\phi(M, M_*) = |H_M|-|H_{M_*} \cap H_M|$. Thus
	from \cref{hweight}, we have that $\phi(M_*, M) - \phi(M, M_*) = |H_{M_*}| - |H_M| \geq 0$.
\end{proof}
\begin{theorem}
	Let $G = (R,B, 2, (\succsim_a)_{a \in R\cup B})$ be a roommate diversity game with room size 2. We can find a popular outcome $\pi$ in polynomial time.
\end{theorem}
\begin{proof}
	We already know that for any (integer) weighted $G_m$, a maximum weight perfect matching can be found in polynomial time\cite{10.1145/3155301}. Thus, using the existing polynomial time algorithm and \cref{roomsize2lem}, for roommate diversity problem with room size 2, a popular outcome always exists and it can be computed in polynomial time.
\end{proof}

%% file: sections/strictpopularity.tex
\section{Strict Popularity}\label{strictpop}
In this section we show that determining the existence of a strictly popular outcome in a roommate diversity game is co-NP-hard. We construct a roommate diversity game $G = (R, B, s, (\succ_a)_{a \in R\cup B})$ with dichotomous preferences from an X3C instance $(X,C)$ in such that there exists a solution $C' \subseteq C$ that partitions $X$ if and only if there exist no popular outcome for $G$.

\subsection{Roommate Diversity Game}\label{mixpoprdgame}
We set the room size $s = 5(q+1)+1+m = 5q+6+m$. Note that $\frac{5(q+1)+1+m}{s} = 1$. The agents and their preference profile is defined as follows.
\begin{table}[h!]
		\begin{tabular}{  r  l  l }
			Set Agents & $R^{set} = \{r_i | i \in X\}= \{r_1, \dots, r_m\}$ & \\
			Redundant Agents & $R^{red}_j = \{r_j^1, \dots, r_j^{5j-2}\}$ & for $j \in [q]$ \\
			 & $R^{red} = R^{red}_1 \cup \dots \cup R^{red}_q$ &  \\
		 	Monolith Agents & $R^{mon} = \{r^1_{mon}, \dots, r^{5(q+1)+1}_{mon}\}$ &
		\end{tabular}
	\caption{Set of red agents $R = R^{set} \cup R^{mon} \cup R^{red}$.}
		\begin{tabular}{  r  l  l }
			Filling Agents & $B^{fill}_j = \{b^1_j, \dots, b^{s-(5j-2)-3}_j\}$ & for $j \in [q]$ \\
		 	& $B^{fill} = B^{fill}_1 \cup \dots \cup B^{fill}_q$ &  \\
			Additional Agents  & $B^{add}_j = \{\tilde{b}^1_j, \tilde{b}^2_j, \tilde{b}^3_j\}$ &  for $j \in [q]$ \\
			 & $B^{add} = B^{add}_1 \cup \dots \cup B^{add}_q$  & \\
			Monolith Agents & $B^{mon} =\{b_{mon}^1, \dots, b_{mon}^{s-(5(q+1)+1)}\}$ & \\
			Evening Agents & $B^{even} = \{b_{even}^1, \dots, b_{even}^{5(q+1)+1}\}$ & 
		\end{tabular}
	\caption{Set of blue agents $B = B^{even} \cup B^{mon} \cup B^{add} \cup B^{fill}$.}
	\begin{center}
				\begin{tabular}{ | l | l | l | l | }
					\hline
					 \multirow{2}{*}{Agent} & \multicolumn{2}{c|}{Preference Profile} & \multirow{2}{*}{} \\\cline{2-3}
					& $D_a^+$ & $D_a^-$ &  \\\hline
					$a = r_i \in R^{set}$ & $\{\frac{5j_1^i+1}{s}, \dots, \frac{5j_{m_i}^i+1}{s}\}  \cup \{\frac{5(q+1)+1+m}{s}\}$ & $D \setminus D_a^+$  & \\  
					$a = r_j^p \in R^{red}_j$ & $\{\frac{5j+1}{s}, \frac{5j-2}{s}\}$ & $D \setminus D_a^+$ & $j \in [q]$ \\
					$a = r_{mon}^{p} \in R^{mon}$ & $\{\frac{5(q+1)+1+m}{s}, \frac{5(q+1)+1}{s}\}$ & $D \setminus D_a^+$ & \\\hline\hline
					$a = b_j^p \in B_j^{fill}$ & $\{\frac{5j+1}{s}, \frac{5j-2}{s}\}$ & $D \setminus D_a^+$ & $j \in [q]$  \\
					$a = \tilde{b}_j^p \in B_j^{add}$ & $\{\frac{5j-2}{s}, 0\}$ & $D \setminus D_a^+$ & $j \in [q]$  \\
					$a = b_{mon}^p \in B^{mon}$ & $\{\frac{5(q+1)+1}{s}, 0\}$ & $D \setminus D_a^+$ & \\
					$a = b_{even}^p \in B^{even}$ & $\{0\}$ & $D \setminus D_a^+$ &  \\\hline
			\end{tabular}
		\caption{Preference profile $(\succsim_a)_a \in R \cup B$.}
	\end{center}
\end{table}
\subsection{Predefined Outcomes}\label{strictpredef}
We define the monolithic outcome $\pi_{mon}$ and for a solution $C' \subseteq C$ for $(X,C)$ the reduced outcome $\pi_{C'}$. In these outcomes every agent is in a room with a fraction that it approves of. Observe that an outcome $\pi$ in which every agent is in a room with a fraction that it approves of is popular, i.e., for any agent $a \in R\cup B$, if $\theta(\pi(a)) \in D_a^+$, then $\pi$ is popular.

\subsubsection{Monolithic Outcome}\label{strictpredef1}
First we define the rooms that contain red agents. Let
\begin{align*}
	& P_j = B_j^{add} \cup R^{red}_j \cup B^{fill}_j \text{ for } j \in [q]; & P_{q+1} = R^{set} \cup R^{mon}.
\end{align*}
Let $\pi_R = \{P_j | j \in [q+1]\}$.
Additionally, let us define the set $B_r$ of remaining blue agents as the blue agents that are not contained in any room $P_j$. That is,
$$B_r = B \setminus \bigcup\limits_{j=1}^{q+1} P_j =  B^{mon} \cup B^{even}.$$
Note that $|B^{mon}| + |B^{even}| = s-(5(q+1)+1) + (5(q+1)+1) = s$. We define $\pi_B$ to be the partition that only contains the room $B_r$. Finally, we define the monolithic outcome to be $\pi_{mon} = \pi_R \cup \pi_B$. Note that every agent is assigned a room by $\pi_{mon}$ with a fraction that it approves of, i.e., for every agent $a \in R\cup B$ we have $\theta(\pi_{mon}(a)) \in D_a^+$. Additionally, a monolithic outcome always exists.

\subsubsection{Reduced Outcome}\label{strictpredef2}
Let the solution $C' \subseteq C$ partition $X$. For every 3-element subset $A_j \in C$, we define the rooms that contain red agents.
\begin{align*}
	& P'_j = \begin{cases}
		\{r_i \in R^{set} | i \in A_j\} \cup R^{red}_j \cup B^{fill}_j & \text{, if } A_j \in C' \\
		B^{add}_j \cup R^{red}_j \cup B^{fill}_j & \text{, if } A_j \notin C'
	\end{cases} \text{ for } j \in [q];\\
	& P'_{q+1} = R^{mon} \cup B^{mon}.
\end{align*}
Let $\pi_R = \{P'_j | j \in [q+1]\}$.
Additionally, let us define the set $B_r$ of remaining blue agents as the blue agents that are not contained in any $P'_j$. That is,
$$B_r = B \setminus \bigcup\limits_{j=1}^{q+1} P'_j = B^{even} \cup \bigcup\limits_{A_j \in C'} B_j^{add}.$$
Since $|\bigcup\limits_{A_j \in C'} B_j^{add}| = \frac{m}{3} \cdot 3 = m$ and $|B^{even}| = 5(q+1)+1$, we have $|B_r| = 5(q+1)+1+m = s$. We define $\pi_B$ to be the partition only containing the coalition $B_r$. 

Finally, we define the reduced outcome to be $\pi_{C'} = \pi_R \cup \pi_B$. Note that every agent is assigned a room by $\pi_{C'}$ with a fraction that it approves of, i.e., for every agent $a \in R\cup B$ we have $\theta(\pi_{C'}(a)) \in D_a^+$. Additionally, a reduced outcome exists if and only if $(X,C)$ has a solution.

\subsection{Hardness}
We demonstrate co-NP-hardness by showing that the monolithic outcome $\pi_{mon}$ is the only popular outcome and therefore strictly popular, if $(X,C)$ has no solution. Otherwise we have multiple popular outcomes, which are exactly the monolithic outcome $\pi_{mon}$ and all the reduced outcomes $\pi_{C'}$, where subset $C' \subseteq C$ is a solution of $(X,C)$.

\begin{lemma}\label{strictpoplemyesno}
	Let $G$ be a roommate diversity game constructed as in \cref{mixpoprdgame} and $\pi$ be an outcome of $G$ such that every agent is in a room with a fraction that it approves of. We have that $\pi$ is either the monolithic outcome $\pi_{mon}$ or a reduced outcome $\pi_{C'}$, where $C' \subseteq C$ partitions $X$. 
\end{lemma}
\begin{proof}
	Let $r_i \in R^{set}$ be an arbitrary set agent. As $r_i$ must be assigned a room in $\pi$ with a fraction that it approves of, we have that $\theta(\pi(r_i)) = 1$ or $\theta(\pi(r_i)) = \frac{5j+1}{s}$, where $j \in J^i$. Let us consider the following cases.
	\begin{enumerate}
		\item $\theta(\pi(r_i)) = 1$. \\
		Since there are exactly $s$ red agents that approve of fraction 1, namely the agents in $R^{set} \cup R^{mon}$, these agents must be contained in the same room, i.e., 
		$$R^{set} \cup R^{mon} \in \pi.$$ 
		A redundant agent $r_{j'}^p \in R^{red}$ cannot be in a room with fraction $\frac{5j'+1}{s}$ as there are only $s-3$ remaining agents that approve of that fraction. Thus any redundant agent $r_{j'}^p$ must be in a room with fraction $\frac{5j'-2}{s}$. The rooms must be $R_{j'}^{red} \cup B_{j'}^{fill} \cup B_{j'}^{add}$ for $j' \in [q]$ as these are exactly the agents that approve of fraction $\frac{5j'-2}{s}$, i.e. for $j'\in [q]$ we have 
		$$R_{j'}^{red} \cup B_{j'}^{fill} \cup B_{j'}^{add} \in \pi.$$ 
		We have $s$ remaining blue agents, namely $B^{even} \cup B^{mon}$, that must belong to the same room , i.e., 
		$$B^{even} \cup B^{mon} \in \pi.$$ 
		Thus $\pi$ must be $\pi_{mon}$ in this case.
		
		\item $\theta(\pi(r_i)) = \frac{5j+1}{s}$. \\
		There are exactly $s$ red agents, including $r_i$, that approve of fraction 1. Since $\theta(\pi(r_i)) \neq 1$ and $\pi$ is an outcome such that every agent is in a room with a fraction that it approves of, the outcome $\pi$ cannot contain a room with only red agents. Thus every set agent $r_{i'} \in R^{set}$ must be in a room with fraction $\frac{5j'+1}{s}$ such that $i' \in A_{j'}$. There are exactly 3 set agents that approve of $\frac{5j''+1}{s}$ for each $j'' \in [q]$. Thus for some solution $C' \subseteq C$ of $(X,C)$, every set agent must be in a room with the shape $\{r_i | i \in A_j \} \cup R^{red}_j \cup B^{fill}_j$ where $A_j \in C'$, i.e., for each $A_j \in C'$ we have
		$$\{r_i | i \in A_j \} \cup R^{red}_j \cup B^{fill}_j \in \pi.$$ 
		The remaining redundant agents $r_{j'''}^p \in R^{red}$ must be in a room with fraction $\frac{5j'''-2}{s}$ as only $s-3$ remaining agents approve of fraction $\frac{5j'''+1}{s}$. The rooms must be $R_{j'''}^{red} \cup B_{j'''}^{fill} \cup B_{j'''}^{add}$ for $A_{{j'''}} \notin C'$, i.e., for each $A_{{j'''}} \notin C'$ we have
		$$R_{j'''}^{red} \cup B_{j'''}^{fill} \cup B_{j'''}^{add} \in \pi.$$ 
		Since there is no room with only red agents, the red monolith agents must be in a room with the blue monolith agents, i.e., 
		$$R^{mon} \cup B^{mon} \in \pi.$$
		We have $s$ remaining blue agents, namely $B^{even} \cup \bigcup\limits_{A_j \in C'}B^{add}_j$, that must belong to the same room, i.e., 
		$$B^{even} \cup \bigcup_{A_j \in C'}B^{add}_j \in \pi.$$ 
		Thus $\pi$ must be $\pi_{C'}$ in this case. \qedhere
	\end{enumerate}
\end{proof}

\begin{theorem}
	Determining whether a strict popular outcome exists in a roommate diversity game is co-NP-hard, even if the preferences are dichotomous.
\end{theorem}
\begin{proof}
	Let us define the family of solution sets $\mathscr{C} = \{C' \subseteq C | C' \text{ is a solution of }(X,C) \}$. By \cref{strictpoplemyesno} and the observation in \cref{strictpredef}, we have that $\mathcal{P} = \{\pi_{C'} | C' \in \mathscr{C}\} \cup \{\pi_{mon}\}$ is the collection of all popular outcomes.
	
	From \cref{strictpoplemyesno} we have that if $(X,C)$ has no solution, then $|\mathcal{P}| = 1$ and therefore we have a strictly popular outcome, namely $\pi_{mon}$. Otherwise we have that $|\mathcal{P}| > 1$ and therefore no outcome is strictly popular.
	
	Thus we have a reduction from any X3C instance $(X,C)$ to a roommate diversity game $G$ with dichotomous preferences such that $(X,C)$ has a solution if and only if $G$ has no strictly popular outcome. Since the X3C problem is NP-complete, determining whether a strict popular outcome exists in a roommate diversity game is co-NP-hard, even if the preferences are dichotomous.
\end{proof}

%% file: sections/mixedpopularity.tex
\section{Mixed Popularity}\label{mixedpop}
In this section we show that a mixed popular outcome is guaranteed to exist in a roommate diversity game. However, computing a mixed popular outcome in polynomial time is not possible unless P=NP. To show the guaranteed existence, we use the minimax theorem\cite{v.Neumann1928}. To show hardness, we construct a reduction from the X3C problem to the roommate diversity problem. The reduction is similar to the reduction in \cref{strictpop}.

\begin{theorem}
	Any roommate diversity game is guaranteed to have a mixed popular outcome.
\end{theorem}
\begin{proof}
	Every roommate diversity game can be viewed as a finite two-player symmetric zero-sum game where the rows and columns of the game matrix are indexed by all possible outcomes $\pi_1, \dots, \pi_{t}$ and the entry at $(i,j)$ of the game matrix has value $\phi(\pi_i,\pi_j)$. By the minimax theorem\cite{v.Neumann1928}, we have that the value of this game is 0. Therefore any maximin strategy is popular.
\end{proof}

We change the reduction in \cref{strictpop} by doubling the agents, changing the room size and preference profile accordingly, and adding new agents. These additions ensure that the monolithic outcome is strictly popular if $(X,C)$ has no solution and that the monolithic outcome is not popular if $(X,C)$ has a solution.

\subsection{Roommate Diversity Game}
We set the room size to $s = 2(5(q + 2) + 1 + m) + 6 = 10q + 28 + 2m$. Note that $\frac{2(5(q + 2) + 1 + m) + 6}{s} = 1$. The agents and their preference profile is defined as follows.
\begin{table}[h!]
	\begin{tabular}{  r  l  l }
		Set Agents & $R^{set} = \{r_i | i \in X\} = \{r_1, \dots, r_m\}$ & \\
		Copy Set Agents & $\tilde{R}^{set} = \{\tilde{r}_i | i \in X\}= \{\tilde{r}_1, \dots, \tilde{r}_m\}$ & \\
		Auxiliary Set Agents & $\hat{R}^{set}= \{\hat{r}_1, \hat{r}_2, \hat{r}_3, \hat{r}_4, \hat{r}_5, \hat{r}_6\}$ & \\
		Redundant Agents & $R^{red}_j = \{r_j^1, \dots, r_j^{2(5j-2)}\}$ & for $j \in [q+1]$ \\
		 & $R^{red} = R^{red}_1 \cup \dots \cup R^{red}_{q+1}$ &  \\
		Monolith Agents & $R^{mon} = \{r^1_{mon}, \dots, r^{2(5(q+2)+1)}_{mon}\}$ & \\
	\end{tabular}
	\caption{Set of red agents $R = R^{set} \cup  \tilde{R}^{set} \cup  \hat{R}^{set} \cup R^{mon} \cup R^{red}$.}
	\begin{tabular}{  r  l  l }
		Filling Agents & $B^{fill}_j = \{b^1_j, \dots, b^{s-2(5j-2)-6}_j\}$ 
		& for $j \in [q+1]$ \\
		& $B^{fill} = B^{fill}_1 \cup \dots \cup B^{fill}_{q+1}$ 
		& \\
		Additional Agents & $B^{add}_j = \{\tilde{b}^1_j, \tilde{b}^2_j, \tilde{b}^3_j, \tilde{b}^4_j, \tilde{b}^5_j, \tilde{b}^6_j\}$ & for $j \in [q+1]$ \\
		& $B^{add} = B^{add}_1 \cup \dots \cup B^{add}_{q+1}$ & \\
		Monolith Agents & $B^{mon} = \{b_{mon}^1, \dots, b_{mon}^{s-2(5(q+2)+1)}\}$ &  \\
		Evening Agents & $B^{even} = \{b_{even}^1, \dots, b_{even}^{2(5(q+2)+1)}\}$ & 
	\end{tabular}
	\caption{Set of blue agents $B = B^{even} \cup B^{mon} \cup B^{add} \cup B^{fill}	$.}
\end{table}
\begin{table}[h!]
	\begin{center}
		\begin{tabular}{ | l | l | l | l | }
			\hline
			\multirow{2}{*}{Agent} & \multicolumn{2}{c|}{Preference Profile} & \multirow{2}{*}{} \\\cline{2-3}
			& $D_a^+$ & $D_a^-$ &  \\\hline
			$a = r_i \in R^{set}$ & $\{\frac{2(5j_1^i+1)}{s}, \dots, \frac{2(5j_{m_i}^i+1)}{s}\}  \cup \{1\}$ & $D \setminus D_a^+$  & \\  
			$a = \tilde{r}_i \in \tilde{R}^{set}$ & $\{\frac{2(5j_1^i+1)}{s}, \dots, \frac{2(5j_{m_i}^i+1)}{s}\}  \cup \{1\}$ & $D \setminus D_a^+$  & \\  
			$a \in \{\hat{r}_1, \dots \hat{r}_5\}$ & $\{\frac{2(5(q+1)+1)}{s}, 1\}$ & $D \setminus D_a^+$  & \\
			$a = \hat{r}_6$ & $\{\frac{2(5(q+1)+1)}{s}\}$ & $D \setminus D_a^+$  & \\  
			$a = r_j^p \in R^{red}_j$ & $\{\frac{2(5j+1)}{s}, \frac{2(5j-2)}{s}\}$ & $D \setminus D_a^+$ & $j \in [q+1]$ \\
			$a = r_{mon}^{p} \in R^{mon}$ & $\{1, \frac{2(5(q+2)+1)}{s}\}$ & $D \setminus D_a^+$ & \\\hline\hline
			$a = b_j^p \in B_j^{fill}$ & $\{\frac{2(5j+1)}{s}, \frac{2(5j-2)}{s}\}$ & $D \setminus D_a^+$ & $j \in [q+1]$  \\
			$a = \tilde{b}_j^p \in B_j^{add}$ & $\{\frac{2(5j-2)}{s}, 0\}$ & $D \setminus D_a^+$ & $j \in [q+1]$  \\
			$a = b_{mon}^p \in B^{mon}$ & $\{\frac{2(5(q+2)+1)}{s}, 0\}$ & $D \setminus D_a^+$ & \\
			$a = b_{even}^p \in B^{even}$ & $\{0\}$ & $D \setminus D_a^+$ &  \\\hline
		\end{tabular}
		\caption{Preference profile $(\succsim_a)_a \in R \cup B$.}
	\end{center}
\end{table}
\newpage
\subsection{Predefined Outcomes}
We define the monolithic outcome $\pi_{mon}$ and for a solution $C' \subseteq C$ for $(X,C)$ the reduced outcome $\pi_{C'}$ in a similar manner as in \cref{strictpredef}. In outcome $\pi_{mon}$ exactly 1 agent is in a room that it disapproves of. In outcome $\pi_{C'}$ all agents are in a room that it approves of.
\subsubsection{Monolithic Outcome}
First we define the rooms that contain red agents.
\begin{align*}
& P_j = B_j^{add} \cup R^{red}_j \cup B^{fill}_j \text{ for } j \in [q+1]; & P_{q+2} = R^{set} \cup \tilde{R}^{set} \cup \hat{R}^{set} \cup R^{mon}.
\end{align*}
Let $\pi_R$ and $\pi_B$ be defined in a similar manner as in \cref{strictpredef1}, i.e., $\pi_R = \bigcup\limits_{j \in [q+2]}\{P_j\}$ and $\pi_B$ contains 1 room with the blue agents not contained in any $P_j$, where $j \in [q+2]$. We define the monolithic outcome to be $\pi_{mon} = \pi_R \cup \pi_B$. Note that agent $\hat{r}_6$ is the only agent that is assigned a room by $\pi_{mon}$ with a fraction that it does not approves of. A monolithic outcome always exists.

\subsubsection{Reduced Outcome}
Let $C' \subseteq C$ be a solution of $(X,C)$, i.e., $C'$ partitions $X$. For every 3-element subset $A_j \in C$, we define the rooms that contain red agents.

\begin{tabular}{l l}
 \multicolumn{2}{l}{$P'_j = \begin{cases}
\{r_i \in R^{set} | i \in A_j\} \cup \{\tilde{r}_i \in \tilde{R}^{set} | i \in A_j\} \cup R^{red}_j \cup B^{fill}_j & \text{, if } A_j \in C' \\
B^{add}_j \cup R^{red}_j \cup B^{fill}_j & \text{, if } A_j \notin C'
\end{cases} \text{ for } j \in [q] ;$}\\ 
 $P'_{q+1} = \hat{R}^{set} \cup R_{q+1}^{red} \cup B_{q+1}^{fill};$ & \\
 $P'_{q+2} = R^{mon} \cup B^{mon}.$
\end{tabular}\\
Let $\pi_R$ and $\pi_B$ be defined in a similar manner as in \cref{strictpredef2}, i.e., $\pi_R = \bigcup\limits_{j \in [q+2]}\{P'_j\}$ and $\pi_B$ contains 1 room with the blue agents not contained in any $P_j'$, where $j\in [q+2]$. We define the reduced outcome to be $\pi_{C'} = \pi_R \cup \pi_B$. Note that every agent is assigned a room by $\pi_{C'}$ with a fraction that it approves of. A reduced outcome exists if and only if $(X,C)$ has a solution.

\subsection{Hardness}
We demonstrate that it is hard to compute a mixed popular solution by showing that $\pi_{mon}$ is strictly popular, if $(X,C)$ has no solution. Otherwise $\pi_{mon}$ is not popular. If $\pi_{mon}$ is strictly popular, then the only mixed popular outcome is $p = \{(\pi_{mon},1)\}$. If $\pi_{mon}$ is not popular, then $p = \{(\pi_{mon},1)\}$ cannot be a mixed popular outcome. 

\begin{observation}\label{mixpopobsmon}
	If outcome $\pi$ has a room that consists of only red agents and assigns exactly 1 agent to a room that it disapproves of, then $\pi$ must be the monolithic outcome. 
	That is, for any $\pi$, if $|D_\pi^-| = 1$ and there exists a room $S \in \pi$ such that $\theta(S)= 1$, then $\pi = \pi_{mon}.$
\end{observation}
\begin{proof}
	Let $\pi$ be an arbitrary outcome that contains a room that consists of only red agents and assigns exactly 1 agent to a room with a fraction that it disapproves of, i.e., for outcome $\pi$ we have $|D_\pi^-| = 1$ and there exists a room $S\in \pi$ such that $\theta(S) = 1$. 
	
	Let $S\in \pi$ be the room that only consists of red agents. We have exactly $s-1$ red agents, namely $R^{set} \cup \tilde{R}^{set} \cup R^{mon} \cup \{\hat{r}_1, \dots, \hat{r}_5\}$, that approve of fraction $1$. Thus the room $S$ must contain $R^{set} \cup \tilde{R}^{set} \cup R^{mon} \cup \{\hat{r}_1, \dots, \hat{r}_5\}$, otherwise there are at least 2 red agents in $S$ that do not approve of fraction $1$. Let $r$ denote the red agent in $S$ that disapproves of fraction $1$. We can write
	$$R^{set} \cup \tilde{R}^{set} \cup R^{mon} \cup \{\hat{r}_1, \dots, \hat{r}_5\} \cup \{r\} \in \pi.$$
	Since $|D_\pi^-| = 1$, $r \in D_\pi^-$, and $r \in S$, any agent in a room in $\pi \setminus \{S\}$ must approve of the fraction of its assigned room.
	Note that $r \in \{\hat{r}_6\} \cup R^{red}$ as these are the red agents that disapprove of fraction $1$. Consider the following cases regarding agent $r$.
	\begin{enumerate}
		\item $r \in R^{red}$.\\
		Then we have that $\hat{r}_6 \notin S$ and $\hat{r}_6 \in D_\pi^+$. Let $S' \in \pi$ denote the room that contains $\hat{r}_6$. Thus we have that $\theta(S') = \frac{2(5(q+1)+1)}{s}$. We have exactly $2(5(q+1)+1)$ red agents that approve of fraction $\frac{2(5(q+1)+1)}{s}$, namely $\hat{R}^{set} \cup R_{q+1}^{red}$. However 5 agents from $\hat{R}^{set}$ are contained in $S$, thus $S'$ must contain at least 5 agents that disapprove of fraction $\frac{2(5(q+1)+1)}{s}$. This contradicts $|D_\pi^-| = 1$, therefore this case cannot occur.
		
		\item $r = \hat{r}_6$.\\
		Then for a redundant red agent $r_j^p \in R^{red}$ we have that $r_j^p \in D_\pi^+$. Let $S' \in \pi$ denote the room that contains $r_j^p$. We have that either $\theta(S') = \frac{2(5j+1)}{s}$ or $\theta(S') = \frac{2(5j-2)}{s}$. 
		
		There are exactly $2(5j+1)$ red agents, namely $R_j^{red} \cup \{r_i, \tilde{r}_i | i \in A_j\}$ or $R_j^{red} \cup \hat{R}^{set}$, that approve of fraction $\frac{2(5j+1)}{s}$ of which 6 are contained in room $S$. Thus if $\theta(S') = \frac{2(5j+1)}{s}$, then $S'$ must contain 6 red agents that disapprove of fraction $\frac{2(5j+1)}{s}$. This contradicts $|D_\pi^-| = 1$, therefore $\theta(S') = \frac{2(5j-2)}{s}$.
		
		There are exactly $2(5j-2)$ red agents that approve of fraction $\frac{2(5j-2)}{s}$, namely $R_j^{red}$. Additionally, there are exactly $s-2(5j-2)$ blue agents that approve of fraction $\frac{2(5j-2)}{s}$, namely $B_j^{add} \cup B_j^{fill}$. Thus $S'$ must contain exactly $R_j^{red} \cup B_j^{add} \cup B_j^{fill}$, as otherwise $S'$ must contain an agent that disapprove of fraction $\frac{2(5j-2)}{s}$. Thus for any $j \in [q+1]$, we have that
		$$R_j^{red} \cup B_j^{add} \cup B_j^{fill} \in \pi.$$
		Since $r = \hat{r}_6$, we also have that
		$$R^{set} \cup \tilde{R}^{set} \cup R^{mon} \cup \hat{R}^{set} \in \pi.$$
		There are exactly $s$ remaining blue agents, namely $B^{mon} \cup B^{even}$. Thus these remaining blue agents must belong to the same room, i.e.,
		$$B^{mon} \cup B^{even} \in \pi.$$
		Thus $\pi$ must be $\pi_{mon}$. \qedhere
	\end{enumerate}
\end{proof}

\begin{lemma}\label{mixpop2lemma}
	If $(X,C)$ has no solution, then for any outcome $\pi$ such that $\pi \neq \pi_{mon}$ has at least 2 agents assigned to a room with a fraction that it does not approve of. That is,
	for any $\pi$ such that $\pi\neq \pi_{mon}$, we have $|D_\pi^-| \geq 2.$
\end{lemma}
\begin{proof}
	To derive a contradiction, assume that there exists an outcome $\pi$ such that $\pi \neq \pi_{mon}$ that assigns fewer than 2 agents to a room with a fraction that it disapproves of, i.e., assume that there exists an outcome $\pi \neq \pi_{mon}$ such that $|D_\pi^-| < 2$. We have the following cases.
	\begin{enumerate}
		\item $|D_\pi^-| = 0$.\label{mixgeq2case1}\\
		As we have exactly $s-1$ red agents that approve of fraction $1$, the outcome $\pi$ cannot have a room that only consists of red agents. Thus a set agent $r_i \in R^{set}$ must be in a room with fraction $\frac{2(5j+1)}{s}$, where $j \in J^i$. Let $S$ denote the room that contains $r_i$.
		
		There are exactly $2(5j+1)$ red agents that approve of fraction $\frac{2(5j+1)}{s}$, namely $R_j^{red} \cup \{r_i, \tilde{r}_i | i \in A_j\}$. Additionally, there are exactly $s-2(5j+1)$ blue agents that approve of fraction $\frac{2(5j+1)}{s}$, namely $B_j^{fill}$. Thus $S$ must contain the agents $\{r_i, \tilde{r}_i | i \in A_j\} \cup R_j^{red} \cup B_j^{fill}$ where $A_j \in C$, i.e.,
		$$\{r_i, \tilde{r}_i | i \in A_j\} \cup R_j^{red} \cup B_j^{fill} \in \pi.$$
		Let $\mathscr{S}$ denote the set of rooms in $\pi$ that contain set agents, i.e.,
		$$\mathscr{S} = \{S' \in \pi | S' \cap R^{set} \neq \emptyset\}.$$
		Let $C'$ be the family of set agent index sets of the rooms in $\mathscr{S}$, i.e.,
		$$C' = \{\{i | r_i \in S'\}  | S' \in \mathscr{S}\}.$$
		
		Since for every set agent $r_{i'} \in R^{set}$, the room in $\pi$ that contains $r_{i'}$ must be of form $\{r_{i'}, \tilde{r}_{i'} | i' \in A_{j'}\} \cup R_{j'}^{red} \cup B_{j'}^{fill}$, we have that $C'$ must be a solution for $(X,C)$. This contradicts $(X,C)$ not having a solution.
		
		\item $|D_\pi^-| = 1$.\\
		As $\pi \neq \pi_{mon}$, by \cref{mixpopobsmon} we have that $\pi$ does not contain a room with only red agents. Let us denote the agent in $D_\pi^-$ by $a$. We have the following 2 cases regarding $a$.
		\begin{enumerate}[\theenumi.1.]
			\item $a \notin R^{set} \cup \tilde{R}^{set}$. \\
			Let $S \in \pi$ denote the room that contains $a$. Let us consider the following cases regarding $S$.
			\begin{enumerate}[\theenumi.\theenumii.1.]
				\item $S \cap (R^{set} \cup \tilde{R}^{set} \cup \hat{R}^{set}) \neq \emptyset$.\\
				Then we have that $\theta(S) = \frac{2(5j+1)}{s}$, where $j \in [q+1]$. 
				
				Consider the following cases regarding $a$.
				\begin{enumerate}[\theenumi.\theenumii.\theenumiii.1.]
					\item $a \in R$.\\
					There are exactly $2(5j+1)$ red agents that approve of fraction $\frac{2(5j+1)}{s}$, namely $\{r_i, \tilde{r}_i | i \in A_j\} \cup R_j^{red}$ or $\hat{R}^{set} \cup R_j^{red}$. Let us denote the set of red agents that approve of fraction $\frac{2(5j+1)}{s}$ by $S^r$, i.e. $S^r = \{r \in R | \frac{2(5j+1)}{s} \in D_a^+\}$.
					
					Room $S$ must contain exactly $2(5j+1)-1$ agents from $S^r$. Let $r \in S^r$ be the agent not contained in $S$ and let $S' \in \pi$ denote the room that contains $r$. We have that $r \in D_\pi^+$, thus either $\theta(S') = \frac{2(5j+1)}{s}$ or $\theta(S') = \frac{2(5j-2)}{s}$.
					
					If $\theta(S') = \frac{2(5j+1)}{s}$, then $S'$ must contain $2(5j+1)-1$ red agents that disapprove of fraction $\frac{2(5j+1)}{s}$. There are a total of $2(5j+1)$ red agents that approve of fraction $\frac{2(5j+1)}{s}$ and $2(5j+1)-1$ of which must be contained in $S$. This would contradict $|D_\pi^-| = 1$.
					
					If $\theta(S') = \frac{2(5j-2)}{s}$, then $S'$ must contain $2(5j-2)-1$ red agents that disapprove of fraction $\frac{2(5j-2)}{s}$. There are a total of $2(5j-2)$ red agents that approve of fraction $\frac{2(5j-2)}{s}$, namely $R_j^{red}$, and $2(5j-2)-1$ of which must be contained in $S$. This would contradict $|D_\pi^-| = 1$.
					
					Thus this case cannot occur.
					
					\item $a \in B$.\\
					Let us denote the set of blue agents that approve of fraction $\frac{2(5j+1)}{s}$ by $S^b$, i.e. $S^b = \{b \in B | \frac{2(5j+1)}{s} \in D_a^+\} = B_j^{fill}$. We have that $|B_j^{fill}| = s-2(5j-2)-6$.
					
					Room $S$ must contain exactly $s-2(5j-2)-7$ agents from $B_j^{fill}$. Let $b \in B_j^{fill}$ be the agent not contained in $S$ and let $S' \in \pi$ denote the room that contains $b$. We have that $b \in D_\pi^+$, thus either $\theta(S') = \frac{2(5j+1)}{s}$ or $\theta(S') = \frac{2(5j-2)}{s}$.
					
					If $\theta(S') = \frac{2(5j+1)}{s}$, then $S'$ must contain $2(5j+1)$ red agents that disapprove of fraction $\frac{2(5j+1)}{s}$. There are a total of $2(5j+1)$ red agents that approve of fraction $\frac{2(5j+1)}{s}$ and all $2(5j+1)$ of them must be contained in $S$. This would contradict $|D_\pi^-| = 1$.
					
					If $\theta(S') = \frac{2(5j-2)}{s}$, then $S'$ must contain $2(5j-2)$ red agents that disapprove of fraction $\frac{2(5j-2)}{s}$. There are a total of $2(5j-2)$ red agents that approve of fraction $\frac{2(5j-2)}{s}$, namely $R_j^{red}$, and all $2(5j-2)$ of them must be contained in $S$. This would contradict $|D_\pi^-| = 1$.
					
					Thus this case cannot occur.
				\end{enumerate}
				
				\item $S \cap (R^{set} \cup \tilde{R}^{set} \cup \hat{R}^{set}) = \emptyset$.\\
				Then every agent in a room that contains a set agent must approve of the fraction of their room, i.e., for any $a' \in R \cup B$, if $\pi(a') \cap R^{set} \neq \emptyset$, then $a \in D_\pi^+$. We have a similar situation as in Case \ref{mixgeq2case1}. 
				
				Every set agent $r_i \in R^{set}$ must be in a room with fraction $\frac{2(5j+1)}{s}$, where $j \in J^i$. Let $S' \in \pi$ denote the room that contains $r_i$. Since every agent in $S'$ must approve of fraction $\frac{2(5j+1)}{s}$, room $S'$ must contain exactly $\{r_i, \tilde{r}_i | i \in A_j\} \cup R_j^{red} \cup B_j^{fill}$ where $A_j \in C$, i.e.,
				$$\{r_i, \tilde{r}_i | i \in A_j\} \cup R_j^{red} \cup B_j^{fill} \in \pi.$$
				Let $\mathscr{S}$ denote the set of rooms in $\pi$ that contain set agents, i.e.,
				$$\mathscr{S} = \{S'' \in \pi | S'' \cap R^{set} \neq \emptyset\}.$$
				Let $C'$ be the family of set agent index sets of the rooms in $\mathscr{S}$, i.e.,
				$$C' = \{\{i | r_i \in S'\}  | S' \in \mathscr{S}\}.$$
				
				Since for every set agent $r_{i'} \in R^{set}$, the room in $\pi$ that contains $r_{i'}$ must be of shape $\{r_{i'}, \tilde{r}_{i'} | i' \in A_{j'}\} \cup R_{j'}^{red} \cup B_{j'}^{fill}$, we have that $C'$ must be a solution for $(X,C)$. This contradicts $(X,C)$ not having a solution.
			\end{enumerate}
			
			\item $a \in R^{set} \cup \tilde{R}^{set}$. \\
			W.l.o.g. assume that $a \in R^{set}$. As $a$ is a set agent, we can write $a = r_i$. Since $|D_\pi^-|=1$, we have that $\tilde{r}_i$ must be in a room $S$ that it approves of, i.e., $\tilde{r}_i \in D_\pi^+$. Let us write $\theta(S) = \frac{2(5j+1)}{s}$, where $j \in J^i$. There are exactly $s$ agents that approve of fraction $\frac{2(5j+1)}{s}$, namely $\{r_{i'}, \tilde{r}_{i'} | i' \in A_j\} \cup R_j^{red}$ which includes $r_i$. Since $r_i \in D_\pi^-$, we have that $r_i$ cannot be in $S$. Thus $S$ must contain an red agent $r$, where $r \neq r_i$, that disapproves of the fraction of fraction $\frac{2(5j+1)}{s}$. Therefore we have at least 2 agents in $D_\pi^-$, i.e., $r_i, r \in D_\pi^-$. This contradicts $|D\pi^-|=1$. \qedhere
		\end{enumerate}
	\end{enumerate}
\end{proof}

\begin{lemma}\label{mixpopnoinst}
	If $(X,C)$ has no solution, then $\pi_{mon}$ is strictly popular.
\end{lemma}
\begin{proof}
	Let $\pi$ be an arbitrary outcome such that $\pi \neq \pi_{mon}$. By \cref{mixpop2lemma} we have that $|D_\pi^-| \geq 2$. Consider the following 2 cases.
	\begin{enumerate}
		\item $\hat{r}_6 \in D_\pi^-$. \\
		Then we have $N(\pi, \pi_{mon}) = \emptyset$. Since $|D_\pi^-| \geq 2$, there exists an agent $a \in D_\pi^-$ such that $a \neq \hat{r}_6$. We have that $a \in D_{\pi_{mon}}^+$, thus $|N(\pi_{mon}, \pi)| \geq 1$. Therefore $\phi(\pi_{mon},\pi) \geq 1$, i.e., $\pi_{mon}$ is more popular than $\pi.$
		
		\item $\hat{r}_6 \notin D_\pi^-$. \\
		Then we have $N(\pi, \pi_{mon}) = \{\hat{r}_6\}$. Since $|D_\pi^-| \geq 2$, there exist agents $a_1,a_2 \in D_\pi^-$ such that $\hat{r}_6 \neq a_1$ and $\hat{r}_6 \neq a_2$. We have that $a_1, a_2 \in D_{\pi_{mon}}^+$, thus $|N(\pi_{mon}, \pi)| \geq 2$. Therefore $\phi(\pi_{mon},\pi) \geq 1$, i.e., $\pi_{mon}$ is more popular than $\pi.$ \qedhere
	\end{enumerate}
\end{proof}

\begin{observation}\label{mixpopyesinst}
	If $(X,C)$ has some solution $C' \subseteq C$, then $\pi_{mon}$ is not popular. This is due to reduced outcome $\pi_{C'}$ being more popular than $\pi_{mon}$.
\end{observation}
\begin{proof}
	The reduced outcome $\pi_{C'}$ is more popular than $\pi_{mon}$. All the agents are assigned a room with a fraction that it approves of by $\pi_{C'}$, i.e., for any $a \in R\cup B$, $\theta(\pi_{C'}(a))\in D_a^+$. Thus we have that $N(\pi_{mon}, \pi_{C'}) = \emptyset$ as no agent can be improved. However, agent $\hat{r}_6$ is better of in $\pi_{C'}$, i.e., $\hat{r}_6 \in D_{\pi_{mon}}^-$ and $\hat{r}_6 \in D_{\pi_{C'}}^+$. Therefore $N(\pi_{C'}, \pi_{mon}) = \{\hat{r}_6\}$ which means that $\pi_{C'}$ is more popular than $\pi_{mon}$. Thus $\pi_{mon}$ is not popular.
\end{proof}
\begin{theorem}
	A mixed popular outcome for a roommate diversity game cannot computed in polynomial time even if the preferences are dichotomous, unless P=NP.
\end{theorem}
\begin{proof}
	From \cref{mixpopnoinst} we have that $\pi_{mon}$ is strictly popular, if $(X,C)$ has no solution. From \cref{mixpopyesinst} we have that $\pi_{mon}$ is not popular, if $(X,C)$ has a solution. Thus if we were able to compute a mixed popular outcome in polynomial time, then we would be able to verify whether $(X,C)$ has a solution in polynomial time by checking whether the computation yields mixed outcome $\{(\pi_{mon}, 1)\}$. If $\{(\pi_{mon}, 1)\}$ is the mixed popular outcome, then $(X,C)$ has no solution. Otherwise $(X,C)$ has a solution. As the X3C problem is NP-complete, this would imply that P=NP.
\end{proof}

%% file: sections/popularity.tex
\section{Popularity}
In this section we show that a popular outcome is not guaranteed to exist in a roommate diversity game. We provide a roommate diversity game in which no popular outcome exist. To show that determining the existence of a popular outcome is co-NP-Hard, we present a reduction from the X3C problem. The reduction is similar to those in the previous sections.

Let $R = \{r_1,r_2,r_3\}$ and $B = \{b_1,b_2,b_3,b_4,b_5,b_6\}$. Consider the roommate diversity game $\overline{G} = (R,B,3,(\succsim_a)_{a\in R\cup B})$ with the following trichotomous preference profiles.

\begin{table}[h!]
	\begin{center}
		\begin{tabular}{ | l | l | l | l | }
			\hline
			\multirow{2}{*}{Agent} & \multicolumn{3}{c|}{Preference Profile}  \\\cline{2-4}
			& $D_a^+$ & $D_a^n$ & $D_a^-$  \\\hline
			$a = r_1$ & $\{\frac{1}{3}\}$ & & $\{\frac{2}{3}, \frac{3}{3}\}$  \\  
			$a \in \{r_2,r_3\}$ & $\{\frac{2}{3}\}$ & & $\{\frac{1}{3}, \frac{3}{3}\}$  \\\hline\hline
			
			$a \in \{b_1,b_2,b_3,b_4\}$ & $\{\frac{1}{3}\}$ & $\{\frac{2}{3}\}$ & $\{\frac{0}{3}\}$ \\
			$a \in \{b_5,b_6\}$ & $\{\frac{0}{3}\}$ & & $\{\frac{1}{3}, \frac{2}{3}\}$\\\hline
		\end{tabular}
		\caption{Preference profile $(\succsim_a)_{a \in R \cup B}$.}
	\end{center}
\end{table}
Let us define the rooms 
\begin{align*}
	& P_1 = \{r_1, \hat{b}_1, \hat{b}_2\}; & P_2 = \{r_2,r_3,\hat{b}_3\}; & & P_3 = \{b_5,b_6,\hat{b}_4\},
\end{align*}
where $\hat{b}_1,\hat{b}_2,\hat{b}_3,\hat{b}_4 \in \{b_1,b_2,b_3,b_4\}$. An outcome $\pi_{top}$ is called a top-type outcome if we can write $\pi_{top} = \{P_1, P_2, P_3\}$.

\begin{lemma}\label{popnotgar2more}
	For any outcome $\pi$ of $\overline{G}$, we have that there are at least 2 agents in a room with a fraction that it does not approve of. That is, for any outcome $\pi$ of $\overline{G}$, 
	$$|D_\pi^n \cup D_\pi^-| \geq 2.$$
\end{lemma}
\begin{proof}
	To derive a contradiction, assume that there exists an outcome $\pi$ of $\overline{G}$ such that $|D_\pi^n \cup D_\pi^-| < 2$. Consider the following cases.
	\begin{enumerate}
		\item $|D_\pi^n \cup D_\pi^-| = 0$.\\
		Then $b_1,b_2,b_3,b_4 \in D_\pi^+$, therefore each of $b_1,b_2,b_3,b_4$ must be in a room with fraction $\frac{1}{3}$. We require at least 2 rooms $S_1,S_2$ with fraction $\frac{1}{3}$ so that each of $b_1,b_2,b_3,b_4$ is contained in a room with fraction  $\frac{1}{3}$. There is exactly 1 red agent that approves of fraction $\frac{1}{3}$, namely $r_1$. Thus $S_1$ or $S_2$ must contain a red agent $r$ that does not approve of fraction $\frac{1}{3}$. Thus $|D_\pi^n \cup D_\pi^-| \geq 1$ in this case as $r \in D_\pi^-$. This contradicts $|D_\pi^n \cup D_\pi^-| = 0$.
		
		\item $|D_\pi^n \cup D_\pi^-| = 1$.\\
		Let us write $D_\pi^n \cup D_\pi^- = \{a\}$. Consider the following cases regarding $a$.
		\begin{enumerate}[\theenumi.1.]
			\item $a = r_1$.\\
			Then we have that $b_1,b_2,b_3,b_4 \in D_\pi^+$, therefore each of $b_1,b_2,b_3,b_4$ must be in a room with fraction $\frac{1}{3}$. We require at least 2 rooms $S_1,S_2$ with fraction $\frac{1}{3}$ so that each of $b_1,b_2,b_3,b_4$ is contained in a room with fraction  $\frac{1}{3}$. There is exactly 1 red agent that approves of fraction $\frac{1}{3}$, namely $r_1$. Thus $S_1$ and $S_2$ both must contain a red agent $r, r' \in R$ that do not approve of fraction $\frac{1}{3}$. Thus $|D_\pi^n \cup D_\pi^-| \geq 2$ in this case as $r,r' \in D_\pi^-$. This contradicts $|D_\pi^n \cup D_\pi^-| = 1$.
			
			\item $a \in \{r_2,r_3\}$.\\
			W.l.o.g. assume that $a = r_2$. We have that $r_3 \in D_\pi^+$, thus the room $S$ that contains $r_3$ must be of fraction $\frac{2}{3}$. There are exactly 2 agents that approve of fraction $\frac{2}{3}$, namely $r_2$ and $r_3$. Since $r_2 \notin S$, as $r_2 \in D_\pi^-$, the room $S$ must contain a red agent $r \neq r_2$ such that $r \in D_\pi^-$. Thus $|D_\pi^n \cup D_\pi^-| \geq 2$ in this case as $r_2,r \in D_\pi^-$. This contradicts $|D_\pi^n \cup D_\pi^-| = 1$.

			\item $a \in \{b_1,b_2,b_3,b_4\}$.\\
			Then we have for each agent $b \in \{b_1,b_2,b_3,b_4\}\setminus \{a\}$, we must have that $b \in D_\pi^+$, therefore $\theta(\pi(b)) = \frac{1}{3}$. We require at least 2 rooms $S_1,S_2$ with fraction $\frac{1}{3}$ so that each agent in $\{b_1,b_2,b_3,b_4\} \setminus \{a\}$ is contained in a room with fraction  $\frac{1}{3}$. There is exactly 1 red agent that approves of fraction $\frac{1}{3}$, namely $r_1$. Thus $S_1$ or $S_2$ must contain a red agent $r$ that does not approve of fraction $\frac{1}{3}$. Note that $a\neq r$ as $a$ must be a blue agent. Thus $|D_\pi^n \cup D_\pi^-| \geq 2$ in this case as $r \in D_\pi^-$ and $a \in D_\pi^n \cup D_\pi^-$. This contradicts $|D_\pi^n \cup D_\pi^-| = 1$.
			
			\item $a \in \{b_5,b_6\}$.\\
			Then $b_1,b_2,b_3,b_4 \in D_\pi^+$, therefore each of $b_1,b_2,b_3,b_4$ must be in a room with fraction $\frac{1}{3}$. We require at least 2 rooms $S_1,S_2$ with fraction $\frac{1}{3}$ so that each of $b_1,b_2,b_3,b_4$ is contained in a room with fraction  $\frac{1}{3}$. There is exactly 1 red agent that approves of fraction $\frac{1}{3}$, namely $r_1$. Thus $S_1$ or $S_2$ must contain a red agent $r$ that does not approve of fraction $\frac{1}{3}$. Note that $r \neq a$ as $a$ must be a blue agent. Thus $|D_\pi^n \cup D_\pi^-| \geq 2$ in this case as $r \in D_\pi^-$ and $a \in D_\pi^n \cup D_\pi^-$. This contradicts $|D_\pi^n \cup D_\pi^-| = 1$.
			\qedhere
		\end{enumerate}
	\end{enumerate}
\end{proof}

\begin{lemma}\label{popnotgartopcard}
	For any outcome $\pi$ of $\overline{G}$, if $|D_\pi^n| = 1$ and $|D_\pi^-| = 1$, then $\pi$ is a top-type outcome. 
\end{lemma}
\begin{proof}
	Let $\pi$ be an outcome such that $|D_\pi^n| = 1$ and $|D_\pi^-| = 1$. Let us write $D_\pi^n = \{a^n\}$ and $|D_\pi^-| = \{a^-\}$. Since the agents in $\{b_1,b_2,b_3,b_4\}$ are the only agents with a non-empty corresponding $D_a^n$, we have that $a^n \in \{b_1,b_2,b_3,b_4\}$. 
	
	To derive a contradiction, assume that $a^- \notin \{b_1,b_2,b_3,b_4\}$. Then we have that $a^- \in \{r_1,r_2,r_3,b_5,b_6\}$. Let us consider the following cases. 
	\begin{enumerate}
		\item $a^- = r_1$.\\
		Then we have that each agent $b \in \{b_1,b_2,b_3,b_4\} \setminus \{a^n\}$ must be in a room that it approves of, i.e., $b \in D_\pi^+$.
		We require at least 2 rooms $S_1,S_2$ with fraction $\frac{1}{3}$ so that each agent in $\{b_1,b_2,b_3,b_4\} \setminus \{a^n\}$ is contained in a room with fraction  $\frac{1}{3}$. There is exactly 1 red agent that approves of fraction $\frac{1}{3}$, namely $r_1$. Thus $S_1$ and $S_2$ both must contain a red agent $r, r' \in R$ that do not approve of fraction $\frac{1}{3}$. Thus $|D_\pi^-| \geq 2$ in this case as $r,r' \in D_\pi^-$. This contradicts $|D_\pi^-| = 1$.

		\item $a^- \in \{r_2,r_3\}$.\\
		W.l.o.g. assume that $a = r_2$. We have that $r_3 \in D_\pi^+$, thus the room $S$ that contains $r_3$ must be of fraction $\frac{2}{3}$. There are exactly 2 agents that approve of fraction $\frac{2}{3}$, namely $r_2$ and $r_3$. Since $r_2 \notin S$, as $r_2 \in D_\pi^-$, the room $S$ must contain a red agent $r \neq r_2$ such that $r \in D_\pi^-$. Thus $|D_\pi^-| \geq 2$ in this case as $r_2,r \in D_\pi^-$. This contradicts $|D_\pi^-| = 1$.

		\item $a^- \in \{b_5,b_6\}$.\\
		Then $b_1,b_2,b_3,b_4 \in D_\pi^+$, therefore each of $b_1,b_2,b_3,b_4$ must be in a room with fraction $\frac{1}{3}$. We require at least 2 rooms $S_1,S_2$ with fraction $\frac{1}{3}$ so that each of $b_1,b_2,b_3,b_4$ is contained in a room with fraction  $\frac{1}{3}$. There is exactly 1 red agent that approves of fraction $\frac{1}{3}$, namely $r_1$. Thus $S_1$ or $S_2$ must contain a red agent $r$ that does not approve of fraction $\frac{1}{3}$. Note that $r \neq a^-$ as $a^-$ must be a blue agent. Thus $|D_\pi^-| \geq 2$ in this case as $r, a^- \in D_\pi^-$. This contradicts $|D_\pi^-| = 1$.
		\qedhere
	\end{enumerate}
\end{proof}

\begin{observation}\label{popnot1}
	For any outcome $\pi$ of $\overline{G}$, we have that $|D_\pi^-| \geq 1$.
\end{observation}
\begin{proof}
	Assume that $|D_\pi^-| = 0$, then we have that $b_5,b_6 \in D_\pi^+$. Therefore there must exist a room $S$ with fraction $\frac{0}{3}$. There are exactly 2 blue agents that approve of fraction $\frac{0}{3}$, the remaining blue agents disapprove of fraction $\frac{0}{3}$. Thus $S$ must contain a blue agent that disapproves of the fraction of its assigned room. Therefore $|D_\pi^-| = 0$ cannot hold. 
\end{proof}

\begin{observation}\label{popnot2}
	For any outcome $\pi$ of $\overline{G}$, if $|D_\pi^n \cup D_\pi^-| = 2$, then $D_\pi^n \cup D_\pi^- \subseteq \{b_1,b_2,b_3,b_4\}$.
\end{observation}
\begin{proof}
	Let $\pi$ be an outcome of $\overline{G}$ such that $|D_\pi^n \cup D_\pi^-| = 2$ and let us write $D_\pi^n \cup D_\pi^- = \{a_1,a_2\}$. To derive a contradiction, assume that $D_\pi^n \cup D_\pi^- \nsubseteq \{b_1,b_2,b_3,b_4\}$. Consider the following cases.
	\begin{enumerate}
		\item $|(D_\pi^n \cup D_\pi^-) \cap \{b_1,b_2,b_3,b_4\}| = 1$.\\
		W.l.o.g. assume that $b_4 \in D_\pi^n \cup D_\pi^-$. We have that $b_1,b_2,b_3 \in D_\pi^+$, therefore each of $b_1,b_2,b_3$ must be in a room with fraction $\frac{1}{3}$. We require at least 2 rooms $S_1,S_2$ with fraction $\frac{1}{3}$ so that each of $b_1,b_2,b_3$ is contained in a room with fraction $\frac{1}{3}$. 
		
		One of $S_1$ or $S_2$ must contain the red agent $r \in \{r_2,r_3\}$. Additionally one of $S_1$ or $S_2$ must contain the blue agent $b \in \{b_5,b_6\}$. Neither agent $r$ nor agent $b$ approves of fraction $\frac{1}{3}$, thus we have that $b_4, r, b \in D_\pi^n \cup D_\pi^-$. This contradicts $|D_\pi^n \cup D_\pi^-| = 2$. Therefore this case cannot occur.

		\item $|(D_\pi^n \cup D_\pi^-) \cap \{b_1,b_2,b_3,b_4\}| = 0$.\\
		Then $b_1,b_2,b_3,b_4 \in D_\pi^+$, therefore each of $b_1,b_2,b_3,b_4$ must be in a room with fraction $\frac{1}{3}$. We require at least 2 rooms $S_1,S_2$ with fraction $\frac{1}{3}$ so that each of $b_1,b_2,b_3,b_4$ is contained in a room with fraction $\frac{1}{3}$. 
		
		Assume that neither $S_1$ nor $S_2$ contain $r_1$. We have that $D_\pi^n \cup D_\pi^- = \{r_2,r_3\}$, thus $r_1 \in D_\pi^+$. The room $S_3$ that contains $r_1$ must contain 2 blue agents $b_5,b_6$ which disapprove of fraction $\frac{1}{3}$. We would have that $D_\pi^n \cup D_\pi^- = \{r_2,r_3,b_5,b_6\}$, which contradicts $|D_\pi^n \cup D_\pi^-| = 2$. Thus one of $S_1$ or $S_2$ must contain $r_1$.
		
		W.l.o.g. assume that $r_2$ is also contained in $S_1$ or $S_2$. We have 3 remaining agents, namely $r_3, b_5, b_6$, thus they must belong in the same room. All 3 agents disapprove of fraction $\frac{1}{3}$, thus we have that $|D_\pi^n \cup D_\pi^-| \geq 3$. This contradicts $|D_\pi^n \cup D_\pi^-| = 2$, thus this case cannot occur.
		\qedhere
	\end{enumerate}
\end{proof}

\begin{lemma}\label{popnotgartopbest}
	Let $\pi$ be an outcome that is not a top-type outcome. There exist a top-type outcome $\pi'$ that is more popular than $\pi$.
\end{lemma}
\begin{proof}
	By \cref{popnotgar2more} we have that $|D_\pi^n \cup D_\pi^-| \geq 2$. By \cref{popnotgartopcard} we have that $|D_\pi^n| \neq 1$ or $|D_\pi^-| \neq 1$. Let us consider the following cases.
	\begin{enumerate}
		\item $|D_\pi^-| \neq 1$.\\
		By \cref{popnot1} and case assumption we have that $|D_\pi^-| \geq 2$. Consider the following cases.
		\begin{enumerate}[\theenumi.1.]
			\item $|D_\pi^-| = 2$.\\
			Consider the following cases.
			\begin{enumerate}[\theenumi.\theenumii.1.]
				\item $|D_\pi^n| = 0$.\\
				By \cref{popnot2} we have that $D_\pi^- \subseteq \{b_1,b_2,b_3,b_4\}$. Let us write $D_\pi^- = \{a_1,a_2\}$. Since $a_1, a_2 \in \{b_1,b_2,b_3,b_4\}$, there exists a top-type outcome $\pi'$ such that $D_{\pi'}^- = \{a_1\}$ and $D_{\pi'}^n = \{a_2\}$. Thus we have that $\phi(\pi',\pi) = 1$.
				
				\item $|D_\pi^n| \geq 1$.\\
				Let us write $D_\pi^- = \{a_1, a_2\}$ and $a_3 \in D_\pi^n$. We have that $a_3 \in \{b_1,b_2,b_3,b_4\}$, since the agents in $\{b_1,b_2,b_3,b_4\}$ are the only agents with a corresponding non-empty $D_a^n$. 
				
				Let $\pi'$ be a top-type outcome such that $\{a_3\} = D_{\pi'}^n$ and $\{a^-\} = D_{\pi'}^-$. Note that $N(\pi, \pi') \subseteq \{a^-\}$. Consider the following cases.
				\begin{enumerate}[\theenumi.\theenumii.\theenumiii.1.]
					\item $N(\pi, \pi') =  \emptyset$.\\
					Then we have that $a^- \in D_\pi^-$. W.l.o.g. assume that $a_1 = a^-$. We have that $a_2 \in D_{\pi'}^+$, thus $a_2 \in N(\pi', \pi)$. Thus we have that $\phi(\pi',\pi) > 0$, i.e., $\pi'$ is more popular than $\pi$.
					
					\item $N(\pi, \pi') = \{a^-\}$.\\
					Then we have that $a^- \in D_\pi^n \cup D_\pi^+$ and $a_1,a_2 \in D_{\pi'}^+$. Thus $a_1,a_2 \in N(\pi',\pi)$. Therefore we have that $\phi(\pi',\pi) > 0$, i.e., $\pi'$ is more popular than $\pi$.
				\end{enumerate}
			\end{enumerate}
			\item $|D_{\pi}^-| > 2$.\\
			Let us write $a_1, a_2, a_3 \in D_\pi^-$. Let $\pi'$ be an arbitrary top-type outcome with $\{a^n\} = D_{\pi'}^n$ and $\{a^-\} = D_{\pi'}^-$. Note that $N(\pi, \pi') \subseteq \{a^n, a^-\}$. Consider the following cases.
			\begin{enumerate}[\theenumi.\theenumii.1.]
				\item $a^- \in  D_\pi^-$.\\
				W.l.o.g. assume that $a_1 = a^-$. In this case, we have that $N(\pi, \pi') \subseteq \{a^n\}$. We have that $a_2, a_3 \in N(\pi', \pi)$. Thus $\phi(\pi', \pi) > 0$. 
				
				\item $a^- \notin  D_\pi^-$.\\
				In this case, we have that $N(\pi, \pi') \subseteq \{a^-, a^n\}$. We have that $a_1, a_2, a_3 \in N(\pi', \pi)$. Thus $\phi(\pi', \pi) > 0$. 
				
			\end{enumerate}
		\end{enumerate}
		
		\item $|D_{\pi}^n| \neq 1$.\\\
		Consider the following cases.
		\begin{enumerate}[\theenumi.1.]
			\item $|D_{\pi}^n| = 0$.\\
			By \cref{popnotgar2more}, we have that $|D_{\pi}^-| \geq 2$. This case is proven similarly to case 1.1.1. as we have the exact same case assumption, i.e., $|D_{\pi}^n| = 0$ and $|D_{\pi}^-| \geq 2$.
			
			\item $|D_{\pi}^n| \geq 2$.\\
			Let us write $a_1, a_2 \in D_\pi^n$. We have that $\{a_1, a_2\} \subseteq \{b_1,b_2,b_3,b_4\}$, since the agents in $\{b_1,b_2,b_3,b_4\}$ are the only agents with a corresponding non-empty $D_a^n$.
			By \cref{popnot1}, we have that $|D_{\pi}^-| \geq 1$. 
			Let us write $a_3 \in D_\pi^-$. Let us consider the following cases regarding $a_3$.
			\begin{enumerate}[\theenumi.\theenumii.1.]
				\item $a_3 \in \{b_1,b_2,b_3,b_4\}$.\\
				By the definition of top-type outcome, we can construct a top-type outcome $\pi'$ such that $D_{\pi'}^- = \{a_3\}$, $D_{\pi'}^n = \{a_2\}$, and $a_1 \in D_{\pi'}^+$. 
				
				We have that $N(\pi, \pi') = \emptyset$ and $|N(\pi',\pi)| \geq 1$. Therefore $\phi(\pi', \pi) > 0$.
				
				\item $a_3 \notin \{b_1,b_2,b_3,b_4\}$.\\
				By the definition of top-type outcome, we can construct a top-type outcome $\pi'$ such that $D_{\pi'}^- = \{a^-\}$, $D_{\pi'}^n = \{a_1\}$, and $a_2, a_3 \in D_{\pi'}^+$. 
				
				We have that $N(\pi,\pi') \subseteq \{a^-\}$. Additionally we have that $a_2, a_3 \in N(\pi',\pi)$, therefore $\phi(\pi', \pi) > 0$. \qedhere
			\end{enumerate}
		\end{enumerate}
	\end{enumerate}
\end{proof}

\begin{lemma}\label{popnotgartopnonpop}
	A top-type outcome $\pi$ is not popular.
\end{lemma}
\begin{proof}
	Let $\pi$ be a top-type outcome. We can write $\pi = \{P_1,P_2,P_3\} = \{\{r_1, \hat{b}_1, \hat{b}_2\}, \{r_2,r_3,\hat{b}_3\},\\ \{b_5,b_6,\hat{b}_4\}\}$, where $\hat{b}_1,\hat{b}_2,\hat{b}_3,\hat{b}_4 \in \{b_1,b_2,b_3,b_4\}$. Let us construct outcome $\pi'$ as follows:
	\begin{align*}
		\pi'=  \{& P_1 \setminus \{\hat{b}_2\} \cup \{\hat{b}_3\}, & P_2 \setminus \{\hat{b}_3\} \cup \{\hat{b}_4\}, && P_3 \setminus \{\hat{b}_4\} \cup \{\hat{b}_2\} \}
	\end{align*}
	We have that $N(\pi,\pi') = \{\hat{b}_2\}$ and $N(\pi,\pi') = \{\hat{b}_3, \hat{b}_4\}$. Thus $\pi'$ is more popular than $\pi$.
\end{proof}

\begin{theorem}
	A popular outcome is not guaranteed to exist in a roommate diversity game $\overline{G}$.
\end{theorem}
\begin{proof}
	From \cref{popnotgartopbest} and \ref{popnotgartopnonpop} we have that no popular outcome exists in $\overline{G}$. Thus not every roommate diversity game permits a popular outcome.
\end{proof}

The idea of the co-NP-hardness reduction is to combine the reduction of \cref{mixedpop} with the proof of a popular outcome not being guaranteed to exist. In this reduction, we have that a popular outcome exists if and only if $(X,C)$ has no solution.

\subsection{Roommate Diversity Game}

We set the room size to $s =2(5(q + 4) + 1) + 2m + 3 = 10q+45+2m$. Note that $\frac{2(5(q + 4) + 1) + 2m + 3}{s} = 1$. The agents and their preference profile is defined as follows.
\begin{table}[h!]
	\begin{tabular}{  r  l  l }
		Circular Set Agents & $\hat{R}^{set}= \{\hat{r}_1, \hat{r}_2, \hat{r}_3\}$ &  \\
		Circular Redundant Agents & $R^{red}_j = \{r_j^1, \dots, r_j^{5j-2}\}$ &  for $j \in [3]$ \\
		Set Agents & $R^{set} = \{r_1, \dots, r_m\}$ &  \\
		Copy Set Agents & $\tilde{R}^{set} = \{\tilde{r}_1, \dots, \tilde{r}_m\}$ & \\
		Redundant Agents & $R^{red}_j = \{r_j^1, \dots, r_j^{2(5j-2)}\}$ & for $j \in [4,q+3]$ \\
		& $R^{red} = R^{red}_1 \cup \dots \cup R^{red}_{q+3}$ & \\
		Monolith Agents & $R^{mon} = \{r^1_{mon}, \dots, r^{s-2m-3}_{mon}\}$ & \\
	\end{tabular}
	\caption{Set of red agents $R = \hat{R}^{set} \cup R^{set} \cup  \tilde{R}^{set} \cup  R^{mon} \cup R^{red}$.}
\begin{tabular}{  r  l  l  }
	Circular Filling Agents & $B^{fill}_j = \{b^1_j, \dots, b^{s-(5j-2)-1}_j\}$ & for $j \in [3]$ \\
	Circular Additional Agents & $B^{add}_j $ $ = \{\tilde{b}_j\}$ & for $j \in [3]$ \\
	Filling Agents & $B^{fill}_j = \{b^1_j, \dots, b^{s-2(5j-2)-6}_j\}$ & for $j \in [4,q+3]$ \\
	& $B^{fill} = B^{fill}_1 \cup \dots \cup B^{fill}_{q+3}$ &  \\
	Additional Agents & $B^{add}_j = \{\tilde{b}^1_j, \tilde{b}^2_j, \tilde{b}^3_j, \tilde{b}^4_j, \tilde{b}^5_j, \tilde{b}^6_j\}$ & for $j \in [4,q+3]$ \\
	& $B^{add} = B^{add}_1 \cup \dots \cup B^{add}_{q+3}$ & \\
	Monolith Agents & $B^{mon} = \{b_{mon}^1, \dots, b_{mon}^{2m+3}\}$ & \\
	Evening Agents & $B^{even} = \{b_{even}^1, \dots, b_{even}^{s-2m-3}\}$ & 
\end{tabular}
\caption{Set of blue agents $B = B^{even} \cup B^{mon} \cup B^{add} \cup B^{fill}$.}
	\begin{center}
		\begin{tabular}{ | l | l | l | l | l | }
			\hline
			\multirow{2}{*}{Agent} & \multicolumn{3}{c|}{Preference Profile} & \multirow{2}{*}{} \\\cline{2-4}
			& $D_a^+$ & $D_a^n$ & $D_a^-$ &  \\\hline
			$a \in \{\hat{r}_1, \hat{r}_2\}$ & $\{\frac{5\cdot(3)-1}{s}\}$ & $\{\frac{5\cdot(2)-1}{s}\}$ & $D \setminus (D_a^+ \cup D_a^n)$  & \\  
			$a \in \hat{r}_3$ & $\{\frac{5\cdot(3)-1}{s}, 1\}$ & $\{\frac{5\cdot(2)-1}{s}\}$ & $D \setminus (D_a^+ \cup D_a^n)$  & \\  
			$a = r_j^p \in R^{red}_j$ & $\{\frac{5j-1}{s}, \frac{5j-2}{s}\}$ & & $D \setminus D_a^+$ & $j \in [2]$ \\
			$a = r_3^p \in R^{red}_3$ & $\{\frac{5\cdot(3)-1}{s}, \frac{5\cdot(3)-2)}{s}\}$ & $\{\frac{5\cdot(2)-1}{s}\}$ & $D \setminus (D_a^+ \cup D_a^n)$ &  \\\hline
			
			$a = r_i \in R^{set}$ & $\{\frac{2(5\tilde{j}_1^i+1)}{s}, \dots, \frac{2(5\tilde{j}_{m_i}^i+1)}{s}\}  \cup \{1\}$ & & $D \setminus D_a^+$  & $\tilde{j}^i_p = j^i_p + 3$ \\  
			$a = \tilde{r}_i \in \tilde{R}^{set}$ & $\{\frac{2(5\tilde{j}_1^i+1)}{s}, \dots, \frac{2(5\tilde{j}_{m_i}^i+1)}{s}\}  \cup \{1\}$ & & $D \setminus D_a^+$  & $\tilde{j}^i_p = j^i_p + 3$ \\   
			$a = r_j^p \in R^{red}_j$ & $\{\frac{2(5j+1)}{s}, \frac{2(5j-2)}{s}\}$ & & $D \setminus D_a^+$ & $j \in [4, q+3]$ \\
			$a = r_{mon}^{p} \in R^{mon}$ & $\{1, \frac{s-2m-3}{s}\}$ & & $D \setminus D_a^+$ & \\\hline\hline
			
			$a = b_j^p \in B_j^{fill}$ & $\{\frac{5j-1}{s}, \frac{5j-2}{s}\}$ & & $D \setminus D_a^+$ & $j \in [3]$  \\
			$a = \tilde{b}_j \in B_j^{add}$ & $\{\frac{5j-2}{s}, 0\}$ & & $D \setminus D_a^+$ & $j \in [3]$  \\\hline
			
			$a = b_j^p \in B_j^{fill}$ & $\{\frac{2(5j+1)}{s}, \frac{2(5j-2)}{s}\}$ & & $D \setminus D_a^+$ & $j \in [4, q+3]$  \\
			$a = \tilde{b}_j^p \in B_j^{add}$ & $\{\frac{2(5j-2)}{s}, 0\}$ & & $D \setminus D_a^+$ & $j \in [4, q+3]$  \\
			$a = b_{mon}^p \in B^{mon}$ & $\{\frac{s-2m-3}{s}, 0\}$ & & $D \setminus D_a^+$ & \\
			$a = b_{even}^p \in B^{even}$ & $\{0\}$ & & $D \setminus D_a^+$ &  \\\hline
		\end{tabular}
		\caption{Preference profile $(\succsim_a)_a \in R \cup B$.}
	\end{center}
\end{table}

\subsection{Predefined Outcomes}
\subsubsection{Monolithic Outcome}
First we define the rooms that contain contain red agents. Let
\begin{align*}
& P_j = B_j^{add} \cup R^{red}_j \cup B^{fill}_j \text{ for } j \in [q+3]; & P_{q+4} = R^{set} \cup \tilde{R}^{set} \cup \hat{R}^{set} \cup R^{mon}.
\end{align*}

Let $\pi_R$ and $\pi_B$ be defined in a similar manner as in \cref{strictpredef1}, i.e., $\pi_R = \bigcup\limits_{j \in [q+4]}\{P_j\}$ and $\pi_B$ contains 1 room with the blue agents not contained in any $P_j$, where $j \in [q+4]$. We define the monolithic outcome to be $\pi_{mon} = \pi_R \cup \pi_B$. Only agents $\hat{r}_1, \hat{r}_2$ disapprove of the fraction of its assigned room. The others approve of the fraction of their assigned room. A monolithic outcome always exists.

\subsubsection{Reduced-type Outcome}
Let the solution $C' \subseteq C$ partition $X$ and consider an arbitrary subset $\{a_1, \dots, a_5\} \subseteq \hat{R}^{set} \cup R_3^{red}$. We define the rooms that contain red agents.

\begin{tabular}{l l}
	$P_j = \{a_j\} \cup R^{red}_j \cup B_j^{fill}$ for $j \in [2]$; & \\
	$P_3 = \{a_3, a_4, a_5\} \cup (R^{red}_3 \setminus \{a_1, \dots, a_5\}) \cup B_3^{fill}$; & \\
	\multicolumn{2}{l}{$P'_{j+3} = \begin{cases}
		\{r_i \in R^{set} | i \in A_j\} \cup \{\tilde{r}_i \in \tilde{R}^{set} | i \in A_j\} \cup R^{red}_{j+3} \cup B^{fill}_{j+3} & \text{, if } A_j \in C' \\
		B^{add}_{j+3} \cup R^{red}_{j+3} \cup B^{fill}_{j+3} & \text{, if } A_j \notin C'
		\end{cases} \text{ for } j \in [q] ;$}\\ 
	 	$P'_{q+4} = R^{mon} \cup B^{mon}.$ & 
\end{tabular}

Let $\pi_R$ and $\pi_B$ be defined in a similar manner as in \cref{strictpredef2}, i.e., $\pi_R = \bigcup\limits_{j \in [q+4]}\{P'_j\}$ and $\pi_B$ contains 1 room with the blue agents not contained in any $P_j'$, where $j\in [q+4]$. We define the reduced-type outcome to be $\pi_{C'} = \pi_R \cup \pi_B$. We have that $D_{\pi_{C'}}^n = \{a_2\}$ and $D_{\pi_{C'}}^- = \{a_1\}$. A reduced-type outcome exists if and only if $(X,C)$ has a solution.
\subsection{Hardness}
We demonstrate co-NP-hardness of determining the existence of a popular outcome in a roommate diversity game by applying similar proof strategies as in our previous proofs. In our reduction, the roommate diversity game has a popular outcome if and only if $(X,C)$ has no solution.

\subsubsection{$(X,C)$ has NO solution}
\begin{lemma}\label{popdpi-2mon}
	Let $\pi$ be an outcome of the roommate diversity game $G$ such that it contains a room with only red agents. If $|D_{\pi}^-| = 2$ and $|D_{\pi}^n| = 0$, then $\pi = \pi_{mon}$. That is,
	$$\forall_{\pi}:(\exists_{S \in \pi}:\theta(S)=1) \wedge |D_{\pi}^-| = 2 \wedge |D_{\pi}^n| = 0 \implies \pi = \pi_{mon}.$$
\end{lemma}
\begin{proof}
	Let $\pi$ be an outcome of the roommate diversity game $G$ such that it contains a room $S^r$ with only red agents and $|D_{\pi}^-| = 2$ and $|D_{\pi}^n| = 0$. Since there are exactly $s-2$ red agents that approve of fraction $1$, $S^r$ must contain $\hat{r}_3$, $R^{set}$, $\tilde{R}^{set}$, $R^{mon}$, and 2 other red agents that disapprove of fraction $1$. Thus we have that $|S^r \cap D_\pi^-| = 2$, i.e., all the agents that disapprove of the fraction of their room are contained in $S^r$.
	
	Since all the blue agents must be in a room in which every agent approves of the fraction of their assigned room, i.e. $\forall_{b\in B}:\forall_{a \in \pi(b)}: \theta(\pi(a)) \in D_a^+$, for $j \in [3,q+3]$ we must have
	$$B_j^{add} \cup R_j^{red} \cup B_j^{fill} \in \pi,$$
	as $b_j^p \in B_j^{fill}$ cannot be in a room with fraction $\frac{2(5j+1)}{s}$, since there are strictly less than $s$ remaining agents that approve of fraction $\frac{2(5j+1)}{s}$.
	
	Note that $S^r$ cannot contain redundant agents from $R^{red}_j$, where $j \in [2]$, as otherwise a redundant agent $r_j^p$ not in $S^r$, i.e. $r_j^p \in R_j^{red} \setminus S^r$, cannot be in a room that only contains agents that approve of the fraction of their assigned room. This would mean that we have $|D_{\pi}^-| > 2$. The only other remaining red agents are $\hat{r}_1$ and $\hat{r}_2$, thus they must be contained in $S^r$. Therefore 
	$$S^r = \hat{R}^{set}\cup R^{set} \cup \tilde{R}^{set} \cup R^{mon} \in \pi.$$
	
	Since all the remaining agents must be in a room with a fraction that they approve of, for $j\in [2]$ we must have
	$$B_j^{add} \cup R_j^{red} \cup B_j^{fill} \in \pi,$$
	and
	$$B^{even} \cup B^{mon}.$$
	Thus we have that $\pi = \pi_{mon}$.
\end{proof}

\begin{lemma}\label{popnogeq2}
	If $(X,C)$ has no solution, then for any outcome $\pi$ of $G$ has at least 2 agents in a room that it disapproves of. That is, for any $\pi$, $|D_\pi^-| \geq 2.$
\end{lemma}
\begin{proof}
	To derive a contradiction, assume that there exists an outcome $\pi$ that assigns fewer than 2 agents to a room that it disapproves of, i.e., assume that there exists an outcome $\pi$ such that $|D_\pi^-| < 2$. Thus we have the following 2 cases.
	\begin{enumerate}
		\item $|D_\pi^-| = 0$.\label{pophardcase1}\\
		As we have exactly $s-2$ red agents that approve of fraction $1$, the outcome $\pi$ cannot have a room that only consists of red agents. Thus every set agent $r_i \in R^{set}$ must be in a room with fraction $\frac{2(5j+1)}{s}$, where $j-3 \in J^i$. However, this means that we can extract a solution $C'$ for $(X,C)$ from $\pi$. This contradicts $(X,C)$ not having a solution.
		
		\item $|D_\pi^-| = 1$.\\
		As we have exactly $s-2$ red agents that approve of fraction $1$ we have that $\pi$ does not contain a room with only red agents. Let us denote the agent in $D_\pi^-$ by $a$. We have the following 2 cases regarding $a$.
		\begin{enumerate}[\theenumi.1.]
			\item $a \notin R^{set} \cup \tilde{R}^{set}$. \\
			Then we have a similar situation as in Case \ref{pophardcase1}. Every set agent $r_i \in R^{set}$ must be in a room with fraction $\frac{2(5j+1)}{s}$, where $j-3 \in J^i$. This means that we can extract a solution $C'$ for $(X,C)$ from $\pi$. Thus this case contradicts $(X,C)$ not having a solution.
			\item $a \in R^{set} \cup \tilde{R}^{set}$. \\
			W.l.o.g. assume that $a \in R^{set}$. As $a$ is a set agent, we can write $a = r_i$. Since there is exactly 1 agent in a room with a fraction that it disapproves of, $\tilde{r}_i$ must be in a room $S$ with a fraction that it approves of. Let us write $\theta(S) = \frac{2(5j+1)}{s}$, where $j-3 \in J^i$. There are exactly $s$ agents that approve of fraction $\frac{2(5j+1)}{s}$, including $r_i$. Since $r_i \in D_\pi^-$, we have that $r_i$ cannot be in $S$. Thus $S$ must contain a red agent other than $r_i$ that disapproves of the fraction of its room. Therefore we have at least 2 agents in $D_\pi^-$. This contradicts $|D\pi^-|=1$. \qedhere
		\end{enumerate}
	\end{enumerate}
\end{proof}

\begin{lemma}\label{lemma0}
	If $(X,C)$ has no solution, then the monolithic outcome $\pi_{mon}$ is popular in the roommate diversity game $G$, i.e., for any outcome $\pi$ we have that
	$$|N(\pi_{mon}, \pi)| \geq |N(\pi, \pi_{mon})|.$$
\end{lemma}
\begin{proof}
	Since $D_{\pi_{mon}}^- = \{\hat{r}_1, \hat{r}_2\}$ and $D_{\pi_{mon}}^n = \emptyset$, for any outcome $\pi \neq \pi_{mon}$ we have that $N(\pi, \pi_{mon}) \subseteq \{\hat{r}_1, \hat{r}_2\}$. Let us consider the following cases regarding $N(\pi, \pi_{mon})$.
	\begin{enumerate}
		\item  $N(\pi, \pi_{mon}) = \emptyset$.\\
		Then $|N(\pi_{mon}, \pi)| \geq |N(\pi, \pi_{mon})|$ trivially holds as $|N(\pi, \pi_{mon})| = 0$.
		
		\item $N(\pi, \pi_{mon}) = \{\hat{r}_1\}$.\\
		Then we have that $\hat{r}_1 \in D_\pi^+ \cup D_\pi^n$ and $\hat{r}_2 \in D_\pi^-$. By \cref{popnogeq2}, we have that $|D_\pi^-| \geq 2$. Therefore there exists $a \in D_\pi^-$ such that $a \neq \hat{r}_2$. Since $a \in D_\pi^-$ we also have $a \neq \hat{r}_1$. Thus $a \in D_{\pi_{mon}}^+$ as $D_{\pi_{mon}}^n = \emptyset$ and $D_{\pi_{mon}}^- = \{\hat{r}_1,\hat{r}_2\}$. Since $a \in D_{\pi_{mon}}^+$ and $a \in D_\pi^-$, we have $a \in N(\pi_{mon}, \pi)$. Therefore $|N(\pi_{mon}, \pi)| \geq 1 = |N(\pi, \pi_{mon})|$.
		
		\item $N(\pi, \pi_{mon}) = \{\hat{r}_2\}$.\\
		Analogous to case 2. We have that $\hat{r}_2 \in D_\pi^+ \cup D_\pi^n$ and $\hat{r}_1 \in D_\pi^-$. By \cref{popnogeq2}, we have that $|D_\pi^-| \geq 2$. Therefore there exists $a \in D_\pi^-$ such that $a \neq \hat{r}_1$. Since $a \in D_\pi^-$ we also have $a \neq \hat{r}_2$. Thus $a \in D_{\pi_{mon}}^+$ as $D_{\pi_{mon}}^n = \emptyset$ and $D_{\pi_{mon}}^- = \{\hat{r}_1,\hat{r}_2\}$. Since $a \in D_{\pi_{mon}}^+$ and $a \in D_\pi^-$, we have $a \in N(\pi_{mon}, \pi)$. Therefore $|N(\pi_{mon}, \pi)| \geq 1 = |N(\pi, \pi_{mon})|$.
		
		\item $N(\pi, \pi_{mon}) = \{\hat{r}_1, \hat{r}_2\}$.\\
		Then we have that $\hat{r}_1, \hat{r}_2 \in D_\pi^+ \cup D_\pi^n$.
		By \cref{popnogeq2}, we have that $|D_\pi^-| \geq 2$. Therefore there exist $a_1, a_2 \in D_\pi^-$ such that $a_1,a_2 \notin \{\hat{r}_1,\hat{r}_2\}$. Thus $a_1,a_2 \in D_{\pi_{mon}}^+$ as $D_{\pi_{mon}}^n = \emptyset$ and $D_{\pi_{mon}}^- = \{\hat{r}_1,\hat{r}_2\}$. Since $a_1,a_2 \in D_{\pi_{mon}}^+$ and $a_1, a_2 \in D_\pi^-$, we have $a_1, a_2 \in N(\pi_{mon},\pi)$. Therefore $|N(\pi_{mon}, \pi)| \geq 2 = |N(\pi, \pi_{mon})|$. \qedhere
	\end{enumerate}
\end{proof}

\subsubsection{$(X,C)$ has a solution}
\begin{lemma}\label{lemma3.1}
	Every outcome $\pi$ has at least 2 agents not in $D_\pi^+$. 
\end{lemma}
\begin{proof}
	Let $\pi$ be an arbitrary outcome. If $\pi$ contains a room that consists only of red agents, then $|D_\pi^- \cup D_\pi^n|\geq 2$ trivially holds for $\pi$ as there are exactly $s-2$ agents that approve of fraction $1$. Thus for the remainder of the proof, we shall assume that $\pi$ does not contain a room that only consists of red agents. That is,
	$$\forall_{S\in \pi}: \theta(\pi(S)) \neq 1.$$
	
	Additionally, if $\pi$ does not contain a room with fraction $\frac{5\cdot(3)-1}{s}$, then the circular set agents $\hat{r}_1, \hat{r}_2, \hat{r}_3$ must be in $D_\pi^n \cup D_\pi^-$ as we assume that $\pi$ has no room consisting of only red agents and $\hat{r}_1, \hat{r}_2, \hat{r}_3$ only approve of fraction $\frac{5\cdot(3)-1}{s}$. Thus $\pi$ must contain a room with fraction $\frac{5\cdot(3)-1}{s}$. That is,
	$$\exists_{S\in \pi}: \theta(\pi(S)) = \frac{5\cdot(3)-1}{s}.$$

	Consider all the agents that approve of fraction $\frac{5\cdot(3)-1}{s}$. That is, let us consider the set 
	$$\{a \in R\cup B | \frac{5\cdot(3)-1}{s} \in D_a^+\} = \hat{R}^{set} \cup R_3^{red} \cup B_3^{fill}.$$ 
	There are a total of $s+2$ agents that approve of fraction $\frac{5\cdot(3)-1}{s}$. Let us denote this set of agents by $H_1 = \hat{R}^{set} \cup R_3^{red} \cup B_3^{fill}.$ 
	Note that for each agent $a$ in $H_1$ their corresponding $D_a^+\setminus \{1\}$ set is a subset of $\{\frac{5\cdot(3)-1}{s}, \frac{5\cdot(3)-2}{s}\}$. That is,
	$$\forall_{a \in H_1}: D_a^+\setminus \{1\} \subseteq \{\frac{5\cdot(3)-1}{s}, \frac{5\cdot(3)-2}{s}\}.$$
	Let us also consider all the agents that approve of fraction $\frac{5\cdot(3)-2}{s}$. That is, let us consider the set 
	$$\{a \in R \cup B | \frac{5\cdot(3)-2}{s} \in D_a^+\} = \{\tilde{b}_3\} \cup R_3^{red} \cup B_3^{fill}.$$ 
	There are a total of $s$ agents that approve of fraction $\frac{5\cdot(3)-2}{s}$. Let us denote this set of agents by $H_2 = \{\tilde{b}_3\} \cup R_3^{red} \cup B_3^{fill}.$ 
	Note that $H_1 \cap H_2 = R_3^{red} \cup B_3^{fill}$ and $|H_1 \cap H_2| = s-1$.
	
	Let $S \in \pi$ be the room such that $\theta(\pi(S)) = \frac{5\cdot(3)-1}{s}$. To derive a contradiction, we assume that $|D_\pi^n \cup D_\pi^-| < 2$. Let us consider the following cases regarding room $S$.
	
	\begin{enumerate}
		\item $|S \cap (D_\pi^n \cup D_\pi^-)| = 1$.\\
		Let $a \in S$ be such that $\frac{5\cdot(3)-1}{s} \notin D_a^+$, i.e., let $a \in S \cap (D_\pi^n \cup D_\pi^-)$. Then we have that $S\setminus \{a\} \subseteq H_1$. Note that 
		$$|H_1 \setminus S| = |H_1| - |S \setminus \{a\}| = 3$$ 
		and 
		$$ |H_2| - |H_1 \cap H_2| = 1 \leq |H_2 \setminus S| \leq  4 = |H_2| - |S \setminus \{a\}| + |\{\hat{r}_1, \hat{r}_2, \hat{r}_3\}|.$$ 
		Since $a \in D_\pi^n \cup D_\pi^-$, for any agent $a' \in H_1 \setminus S$ we have that $a' \in D_\pi^+$. As $\pi$ does not contain a room that consists of only red agents and $a' \in H_1$, we have that $\theta(\pi(a')) \in \{\frac{5\cdot(3)-1}{s}, \frac{5\cdot(3)-2}{s}\}$. For an arbitrary agent $a' \in H_1 \setminus S$, let us consider the following cases.
		
		\begin{enumerate}[\theenumi.1.]
			\item $\theta(\pi(a')) = \frac{5\cdot(3)-1}{s}$.\\
			We have at least $s - |H_1\setminus S| = s - 3 \geq 2$ agents in room $\pi(a')$ that do not approve of fraction $\frac{5\cdot(3)-1}{s}$. This is due to the fact that the only remaining agents outside of $S$ that approve of fraction $\frac{5\cdot(3)-1}{s}$, are the agents in $H_1 \setminus S$. Thus in this case, we have that $|D_\pi^n \cup D_\pi^-| \geq 2$.
			\item  $\theta(\pi(a')) = \frac{5\cdot(3)-2}{s}$.\\
			We have at least $s - |H_2\setminus S| \geq s - 4 > 2$ agents in room $\pi(a')$ that do not approve of fraction $\frac{5\cdot(3)-2}{s}$.  This is due to the fact that the only remaining agents outside of $S$ that approve of fraction $\frac{5\cdot(3)-2}{s}$, are the agents in $H_2 \setminus S$. Thus in this case, we have that $|D_\pi^n \cup D_\pi^-| \geq 2$.
		\end{enumerate}
		\item $|S \cap (D_\pi^n \cup D_\pi^-)| = 0$.\\
		Then we have that $S \subseteq H_1$. Note that 
		$$|H_1 \setminus S| = |H_1| - |S| = 2$$ 
		and 
		$$|H_2| - |H_1 \cap H_2| = 1 \leq |H_2 \setminus S| \leq 3 = |H_2| - |S| + |\{\hat{r}_1, \hat{r}_2, \hat{r}_3\}|.$$ 
		Let us denote $H_1 \setminus S = \{a_1, a_2\}$. Since we assume that $|D_\pi^n \cup D_\pi^-| <2$, at least one of $a_1$ and $a_2$ must be in $D_\pi^+$. That is,
		$$a_1 \in D_{\pi}^+ \vee a_2 \in D_{\pi}^+.$$
		
		Let $i \in [2]$ such that $a_i \in D_{a_i}^+$. Then we have that $\theta(\pi(a_i)) \in \{\frac{5\cdot(3)-1}{s}, \frac{5\cdot(3)-2}{s}\}$, since $a_i \in H_1$ and $\pi$ has no room that consists of only red agents. Consider the following cases.
		
		\begin{enumerate}[\theenumi.1.]
			\item $\theta(\pi(a_i)) = \frac{5\cdot(3)-1}{s}$.\\
			We have at least $s - |H_1 \setminus S| = s - 2 > 2$ agents in room $\pi(a_i)$ that do not approve of fraction $\frac{5\cdot(3)-1}{s}$. This is due to the fact that the only remaining agents outside of $S$ that approve of fraction $\frac{5\cdot(3)-1}{s}$ are the agents in $H_1 \setminus S$. Thus in this case, we have that  $|D_\pi^n \cup D_\pi^-| \geq 2$.
			\item $\theta(\pi(a_i)) = \frac{5\cdot(3)-2}{s}$.\\
			We have at least $s - |H_2\setminus S| \geq s - 3 > 2$ agents in room $\pi(a_i)$ that do not approve of fraction $\frac{5\cdot(3)-2}{s}$.  This is due to the fact that the only remaining agents outside of $S$ that approve of fraction $\frac{5\cdot(3)-2}{s}$ are the agents in $H_2 \setminus S$. Thus in this case, we have that  $|D_\pi^n \cup D_\pi^-| \geq 2$. \qedhere
		\end{enumerate}
	\end{enumerate}
\end{proof}

\begin{lemma}\label{lemma3.6}
	If for an outcome $\pi$ we have that $|D_\pi^n \cup D_\pi^-| = 2$, then $D_\pi^n \cup D_\pi^- \subseteq \hat{R}^{set} \cup R_3^{red}$. 
\end{lemma}
\begin{proof}
	Let $\pi$ be an outcome such that $|D_\pi^n \cup D_\pi^-| = 2$. If $\pi$ contains a room consisting only of red agents, then by \cref{popdpi-2mon} $\pi = \pi_{mon}$. By definition of $\pi_{mon}$, we have $D_\pi^n \cup D_\pi^- \subseteq \hat{R}^{set} \cup R_3^{red}$. Therefore the lemma trivially holds, when $\pi$ contains a room consisting of only red agents. Thus for the remainder of the proof, we shall consider the case where $\pi$ does not contain a room that consists of only red agents.
	
	To derive a contradiction, let us assume that $D_\pi^n \cup D_\pi^- \nsubseteq \hat{R}^{set} \cup R_3^{red}$. This means that at most 1 agent from $\hat{R}^{set} \cup R_3^{red}$ is contained in $D_\pi^n \cup D_\pi^-$. Thus we can consider the following cases.
	\begin{enumerate}
		\item $|(D_\pi^n \cup D_\pi^-) \cap (\hat{R}^{set} \cup R_3^{red})| = 1$.\\
		Let $\{a\} = (D_\pi^n \cup D_\pi^-) \cap (\hat{R}^{set} \cup R_3^{red})$. The remaining agents in $(\hat{R}^{set} \cup R_3^{red}) \setminus \{a\}$ must be contained in $D_\pi^+$. Since $\pi$ does not contain a room consisting of only red agents, the remaining agents in $(\hat{R}^{set} \cup R_3^{red}) \setminus \{a\}$ must be in a room with fraction $\frac{5\cdot 3 - 1}{s}$ or $\frac{5\cdot 3 - 2}{s}$. 
		
		There are $15$ agents in $(\hat{R}^{set} \cup R_3^{red}) \setminus \{a\}$. The remaining red agents that approve of fraction $\frac{5\cdot 3 - 1}{s}$ and/or $\frac{5\cdot 3 - 2}{s}$ are exactly $(\hat{R}^{set} \cup R_3^{red}) \setminus \{a\}$. Thus the agents $(\hat{R}^{set} \cup R_3^{red}) \setminus \{a\}$ must be contained in at least 2 separate rooms $S_1, S_2$ with fraction $\frac{5\cdot 3 - 1}{s}$ or $\frac{5\cdot 3 - 2}{s}$. The blue agents that approve of fraction $\frac{5\cdot 3 - 1}{s}$ or $\frac{5\cdot 3 - 2}{s}$ are exactly the agents in $B_3^{fill} \cup \{\tilde{b}_3\}$ of which there are exactly $s-(5\cdot(3)-2)$. Thus we have at least $s-5\cdot 3 -1$ blue agents in $S_1 \cup S_2$ such that they are contained in $D_\pi^-$. Therefore this case contradicts our assumption that $|D_\pi^n \cup D_\pi^-| = 2$.

		\item $|(D_\pi^n \cup D_\pi^-) \cap (\hat{R}^{set} \cup R_3^{red})| = 0$.\\
		Then all the agents in $\hat{R}^{set} \cup R_3^{red}$ must be contained in $D_\pi^+$. Since $\pi$ does not contain a room consisting of only red agents, the remaining agents of $\hat{R}^{set} \cup R_3^{red}$ must be in a room with fraction $\frac{5\cdot 3 - 1}{s}$ or $\frac{5\cdot 3 - 2}{s}$.  
		
		There are $16$ remaining agents in $\hat{R}^{set} \cup R_3^{red}$. The remaining red agents that approve of fraction $\frac{5\cdot 3 - 1}{s}$ and/or $\frac{5\cdot 3 - 2}{s}$ are exactly $\hat{R}^{set} \cup R_3^{red}$. Thus the agents $\hat{R}^{set} \cup R_3^{red}$ must be contained in at least 2 separate rooms $S_1, S_2$ with fraction $\frac{5\cdot 3 - 1}{s}$ or $\frac{5\cdot 3 - 2}{s}$. The blue agents that approve of fraction $\frac{5\cdot 3 - 1}{s}$ or $\frac{5\cdot 3 - 2}{s}$ are exactly the agents in $B_3^{fill} \cup \{\tilde{b}_3\}$ of which there are exactly $s-(5\cdot(3)-2)$. Thus we have at least $s-5\cdot 3 -1$ blue agents in $S_1 \cup S_2$ such that they are contained in $D_\pi^-$. Therefore this case contradicts our assumption that $|D_\pi^n \cup D_\pi^-| = 2$.
	\end{enumerate}
	Thus our assumption that $D_\pi^n \cup D_\pi^- \nsubseteq \hat{R}^{set} \cup R_3^{red}$ cannot hold. Therefore $D_\pi^n \cup D_\pi^- \subseteq \hat{R}^{set} \cup R_3^{red}$.
\end{proof}

\begin{lemma}\label{lemma3.2}
	Every outcome $\pi$ has at least 1 agents in $D_\pi^-$. 
\end{lemma}
\begin{proof}
	Let $\pi$ be an arbitrary outcome. By \cref{lemma3.1}, we have that there are at least 2 agents in $D_\pi^n \cup D_\pi^-$.
	
	To derive a contradiction assume that there exists an outcome $\pi$ such that $D_\pi^n \cup D_\pi^- = D_\pi^n$. That is, we assume that $D_\pi^- = \emptyset$. Note that $\pi$ cannot contain a room that only consists of red agents as there are exactly $s-2$ agents that do not disapprove of fraction $1$. By construction of the roommate diversity game, we have that $D_\pi^n \cup D_\pi^- \subseteq \hat{R}^{set} \cup R_3^{red}$, as these are the only agents with non-empty $D_a^n$ sets. The only fraction in the set $D_a^n$ is $\frac{5\cdot(2)-1}{s}$. Note that $2 \leq |D_\pi^n \cup D_\pi^-| = |D_\pi^n|  \leq 5\cdot (3)+1 = 16$.
	
	Consider all the agents that do not disapprove of fraction $\frac{5\cdot(2)-1}{s}$. That is, let us consider the set 
	$$\{a \in R\cup B | \frac{5\cdot(2)-1}{s} \notin D_a^- \} = \hat{R}^{set} \cup R_2^{red} \cup R_3^{red} \cup B_2^{fill}.$$ 
	Let us consider the following cases regarding the cardinality of $D_\pi^n \cap \hat{R}^{set}$.
	
	\begin{enumerate}
		\item $|D_\pi^n \cap \hat{R}^{set}| = 3$.\\
		Note that $\hat{r}_1, \hat{r}_2 , \hat{r}_3$ must be in the same room $S$. If they are in 2 separate rooms, we require $2 \cdot (s - (5 \cdot 2 -1))$ blue agents that do not disapprove of fraction $\frac{5\cdot(2)-1}{s}$ and we only have  $s - (5 \cdot 2 -1)$ blue agents, namely $B_2^{fill}$. Similarly, if they are in 3 separate rooms, we require $3 \cdot (s - (5 \cdot 2 -1))$ blue agents that do not disapprove of fraction $\frac{5\cdot(2)-1}{s}$.
		
		Since $D_\pi^- = \emptyset$ and $\theta(\pi(S)) = \frac{5\cdot(2)-1}{s}$, all the blue agents in $B_2^{fill}$ must be contained in $S$.
		
		Since $S$ contains $5\cdot(2)-1$ red agents, 3 of which are $\hat{r}_1, \hat{r}_2 , \hat{r}_3$, there must exist a circular agent $r_2^i \in R_2^{red}$, where $1 \leq i \leq 5\cdot 2 -2$ such that $r_2^i \notin S$. Let $S'$ denote the room in $\pi$ that contains $r_2^i$. Since $D_\pi^- = \emptyset$, we have that $r_2^i \in D_\pi^+$. Thus $\theta(S') = \frac{5\cdot 2 -1}{s}$ or $\theta(S') = \frac{5\cdot 2 -2}{s}$. There are exactly $s-(5\cdot 2-2)$ blue agents that do not disapprove of fraction $\frac{5\cdot 2 -1}{s}$ and/or $\frac{5\cdot 2 -2}{s}$, namely $B_2^{fill} \cup B_2^{add}$. Since $B_2^{fill} \subseteq S$, there is at least 1 blue agent in $S'$ that is contained in $D_\pi^-$, which contradicts our assumption that $D_\pi^- = \emptyset$.
		
		\item $|D_\pi^n \cap \hat{R}^{set}| = 2$.\\
		W.l.o.g. let $\hat{r}_1, \hat{r}_2  \in D_\pi^n$, thus $\hat{r}_3 \in D_\pi^+$. Note that $\hat{r}_1, \hat{r}_2 $ must be in the same room $S$. If they are in 2 separate rooms, we require $2 \cdot (s - (5 \cdot 2 -1))$ blue agents that do not disapprove of fraction $\frac{5\cdot(2)-1}{s}$ and we only have  $s - (5 \cdot 2 -1)$ blue agents, namely $B_2^{fill}$.
		
		Since $D_\pi^- = \emptyset$ and $\theta(\pi(S)) = \frac{5\cdot(2)-1}{s}$, all the blue agents in $B_2^{fill}$ must be contained in $S$.
		
		Since $S$ contains $5\cdot(2)-1$ red agents, 2 of which are $\hat{r}_1, \hat{r}_2 $, there must exist a circular agent $r_2^i \in R_2^{red}$, where $1 \leq i \leq 5\cdot 2 -2$ such that $r_2^i \notin S$. Let $S'$ denote the room in $\pi$ that contains $r_2^i$. Since $D_\pi^- = \emptyset$, we have that $r_2^i \in D_\pi^+$. Thus $\theta(S') = \frac{5\cdot 2 -1}{s}$ or $\theta(S') = \frac{5\cdot 2 -2}{s}$. There are exactly $s-(5\cdot 2-2)$ blue agents that do not disapprove of fraction $\frac{5\cdot 2 -1}{s}$ and/or $\frac{5\cdot 2 -2}{s}$, namely $B_2^{fill} \cup B_2^{add}$. Since $B_2^{fill} \subseteq S$, there is at least 1 blue agent in $S'$ that is contained in $D_\pi^-$, which contradicts our assumption that $D_\pi^- = \emptyset$
		
		\item $|D_\pi^n \cap \hat{R}^{set}| = 1$.\\
		W.l.o.g. let $\hat{r}_1 \in D_\pi^n$, thus $\hat{r}_2, \hat{r}_3 \in D_\pi^+$. Note that $\hat{r}_2, \hat{r}_3$ must be in the same room $S$. If they are in 2 separate rooms, we require $2 \cdot (s - (5 \cdot 3 -1))$ blue agents that do not disapprove of fraction $\frac{5\cdot(3)-1}{s}$ and we only have $s - (5 \cdot 3 -1)$ blue agents, namely $B_3^{fill}$.
		
		Since $D_\pi^- = \emptyset$ and $\theta(\pi(S)) = \frac{5\cdot(3)-1}{s}$, all the blue agents in $B_3^{fill}$ must be contained in $S$.
		
		Since $S$ contains $5\cdot(3)-1$ red agents, 2 of which are $\hat{r}_2 , \hat{r}_3$ and the remaining $5\cdot(3)-1-2$ from $R_3^{red}$, there must exists a circular agent $r_3^i \in R_3^{red}$, where $1 \leq i \leq 5\cdot 3 -2$ such that $r_3^i \notin S$. Let $S'$ denote the room in $\pi$ that contains $r_3^i$. Since $D_\pi^- = \emptyset$, we have that $r_3^i \in D_\pi^+ \cup D_\pi^n$. Thus $\theta(S') = \frac{5\cdot 3 -1}{s}$, $\theta(S') = \frac{5\cdot 3 -2}{s}$ or $\theta(S') = \frac{5\cdot 2 -1}{s}$. The blue agents that do not disapprove of fraction $\frac{5\cdot 3 -1}{s}$ and/or $\frac{5\cdot 3 -2}{s}$ are exactly $B_3^{fill}\cup B_3^{add}$ of which all but 1 are contained in $S$. Thus we have that $\theta(S') = \frac{5\cdot 2 -1}{s}$, otherwise $D_\pi^-$ cannot be empty.
		
		Since $D_\pi^- = \emptyset$ and $\theta(\pi(S')) = \frac{5\cdot(2)-1}{s}$, all the blue agents in $B_2^{fill}$ must be contained in $S'$.
		
		Since $S'$ must contain $5\cdot(2)-1$ red agents from $\hat{r}_1$, $R_3^{red}$ and $R_2^{red}$, as these are the only remaining red agents that do not disapprove of fraction $\frac{5\cdot(2)-1}{s}$, there must exist an agent $r \in \{\hat{r}_1\} \cup R_3^{red} \cup R_2^{red}$ such that $r \notin S'$. Let $S''$ denote the room in $\pi$ that contains $r$. Since $D_\pi^- = \emptyset$, we have that $r \in D_\pi^+ \cup D_\pi^n$. Thus $\theta(S'') = \frac{5\cdot 2 - 1}{s}$, $\theta(S'') = \frac{5\cdot 2 - 2}{s}$, $\theta(S'') = \frac{5\cdot 3 - 1}{s}$, or $\theta(S'') = \frac{5\cdot 3 - 2}{s}$. The blue agents that do not disapprove of fraction $\frac{5\cdot 3 -1}{s}$ and/or $\frac{5\cdot 3 -2}{s}$ are exactly in $B_3^{fill} \cup B_3^{add}$ of which all but 1 are contained in $S$. The blue agents that do not disapprove of fraction $\frac{5\cdot 2 -1}{s}$ and/or $\frac{5\cdot 2 -2}{s}$ are exactly $B_2^{fill} \cup B_2^{add}$ of which all but 1 are contained in $S'$. Thus at least 1 blue agent in $S''$ must be contained in $D_\pi^-$, which contradicts our assumption that $D_\pi^- = \emptyset$.
		
		\item $|D_\pi^n \cap \hat{R}^{set}| = 0$.\\
		Note that $\hat{r}_1, \hat{r}_2 , \hat{r}_3 \in D_\pi^+$ and that $\hat{r}_1, \hat{r}_2 , \hat{r}_3$ must be in the same room $S$. If they are in 2 separate rooms, we require $2 \cdot (s - (5 \cdot 3 -1))$ blue agents that do not disapprove of fraction $\frac{5\cdot(3)-1}{s}$ and we only have $s - (5 \cdot 3 -1)$ blue agents, namely $B_3^{fill}$. Similarly, if they are in 3 separate rooms, we require $3 \cdot (s - (5 \cdot 3 -1))$ blue agents that do not disapprove of fraction $\frac{5\cdot(3)-1}{s}$.
		
		Since $D_\pi^- = \emptyset$ and $\theta(\pi(S)) = \frac{5\cdot(3)-1}{s}$, all the blue agents in $B_3^{fill}$ must be contained in $S$.
		
		Since $S$ contains $5\cdot(3)-1$ red agents, 3 of which are $\hat{r}_1, \hat{r}_2 , \hat{r}_3$ and the remaining $5\cdot(3)-1-3$ from $R_3^{red}$, there must exist a circular agent $r_3^i \in R_3^{red}$, where $1 \leq i \leq 5\cdot 3 -2$ such that $r_3^i \notin S$. Let $S'$ denote the room in $\pi$ that contains $r_3^i$. Since $D_\pi^- = \emptyset$, we have that $r_3^i \in D_\pi^+ \cup D_\pi^n$. Thus $\theta(S') = \frac{5\cdot 3 -1}{s}$, $\theta(S') = \frac{5\cdot 3 -2}{s}$ or $\theta(S') = \frac{5\cdot 2 -1}{s}$. The blue agents that do not disapprove of fraction $\frac{5\cdot 3 -1}{s}$ and/or $\frac{5\cdot 3 -2}{s}$ are exactly $B_3^{fill} \cup B_3^{add}$ of which all but 1 are contained in $S$. Thus we have that $\theta(S') = \frac{5\cdot 2 -1}{s}$, otherwise $D_\pi^-$ cannot be empty.
		
		Since $D_\pi^- = \emptyset$ and $\theta(\pi(S')) = \frac{5\cdot(2)-1}{s}$, all the blue agents in $B_2^{fill}$ must be contained in $S'$.
		
		Since $S'$ must contain $5\cdot(2)-1$ red agents from $R_3^{red}$ and $R_2^{red}$, as these are the only remaining red agents that do not disapprove of fraction $\frac{5\cdot(2)-1}{s}$, there must exist an agent $r_j^k \in R_3^{red} \cup R_2^{red}$ where $2 \leq j \leq 3$ and $1 \leq k \leq 5 \cdot j - 2$ such that $r_j^k \notin S'$. Let $S''$ denote the room in $\pi$ that contains $r_j^k$. Since $D_\pi^- = \emptyset$, we have that $r_j^k \in D_\pi^+ \cup D_\pi^n$. Thus $\theta(S'') = \frac{5\cdot 2 - 1}{s}$, $\theta(S'') = \frac{5\cdot 2 - 2}{s}$, $\theta(S'') = \frac{5\cdot 3 - 1}{s}$, or $\theta(S'') = \frac{5\cdot 3 - 2}{s}$. The blue agents that do not disapprove of fraction $\frac{5\cdot 3 -1}{s}$ and/or $\frac{5\cdot 3 -2}{s}$ are exactly $B_3^{fill}\cup B_3^{add}$ of which all but 1 are contained in $S$. The blue agents that do not disapprove of fraction $\frac{5\cdot 2 -1}{s}$ and/or $\frac{5\cdot 2 -2}{s}$ are exactly $B_2^{fill} \cup B_2^{add}$ of which all but 1 are contained in $S'$. Thus at least 1 blue agent in $S''$ must be contained in $D_\pi^-$, which contradicts our assumption that $D_\pi^- = \emptyset$. \qedhere
	\end{enumerate}
\end{proof}

\begin{lemma}\label{lemma3.7}
	Let $\pi$ be an outcome of the roommate diversity game and assume that $(X,C)$ has a solution. If for outcome $\pi$ it holds that $|D_\pi^n|=1$ and $|D_\pi^-|=1$, then $\pi$ must be a reduced-type outcome.
\end{lemma}

\begin{proof}
	Let $\pi$ be an outcome such that $|D_\pi^n|=1$ and $|D_\pi^-|=1$. By \cref{lemma3.6}, we have that $D_\pi^n \cup D_\pi^- \subseteq \hat{R}^{set} \cup R_3^{red}$. \cref{lemma3.6} also implies that $B \subseteq D_\pi^+$. Let us denote the agents in $D_\pi^n$ and $D_\pi^-$ by $r^n$ and $r^-$ respectively. Note that $\pi$ cannot contain a room that only consists of red agents, as there are exactly $s-2$ red agents that do not disapprove of fraction $1$. Let $C' \subseteq C$ denote some solution of $(X,C)$. We shall show that the following 7 properties hold for $\pi$.
	\begin{enumerate}
		\item\label{lem25prop1} $R^{mon} \cup B^{mon} \in \pi$.\\
		Let $S^{mon} \in \pi$ denote a room that contains some monolith red agent. Since $D_\pi^n \cup D_\pi^- \subseteq \hat{R}^{set} \cup R_3^{red}$, all the monolith red agents $R^{mon}$ and monolith blue agents $B^{mon}$ must be contained in $D_\pi^+$. Since $\pi$ cannot contain a room consisting of only red agents, any monolith red agent must be in a room of fraction $\frac{s-2m-3}{s}$. The only agents that approve of fraction $\frac{s-2m-3}{s}$ are exactly the monolith red agents $R^{mon}$ and monolith blue agents $B^{mon}$, any other agent disapproves of that fraction.
		
		There are exactly $s-2m-3$ monolith red agents. Thus we have that $R^{mon} \subseteq S^{mon}$. If there exists a monolith red agent $r^p_{mon} \in R^{mon}$ such that $r^p_{mon} \notin S^{mon}$, we would require 2 rooms with fraction $\frac{s-2m-3}{s}$. This would contradict $|D_\pi^-|=1$ as we would have at least $s-2m-3$ red agents in $D_\pi^-$.
		
		There are exactly $2m+3$ monolith blue agents. Thus we have that $B^{mon} \subseteq S^{mon}$. If there exists a monolith blue agent $b^p_{mon} \in B^{mon}$ such that $b^p_{mon} \notin S^{mon}$, there must be a blue agent $b \in S^{mon}$ that disapproves of fraction $\frac{s-2m-3}{s}$. This would contradict $D_\pi^n \cup D_\pi^- \subseteq \hat{R}^{set} \cup R_3^{red}$, since we would have $b \in D_\pi^-$.
		
		Thus the agents in $R^{mon}$ and $B^{mon}$ must belong to the same room, i.e.,
		$$S^{mon} = R^{mon} \cup B^{mon} \in \pi.$$
		
		\item\label{lem25prop2} $\{r_i, \tilde{r}_i\}_{i \in A_{j-3}} \cup R_j^{red} \cup B_j^{fill} \in \pi$ for each $A_{j-3} \in C'$.\\
		Since $\pi$ does not contain a room that only consists of red agents and $D_\pi^n \cup D_\pi^- \subseteq \hat{R}^{set} \cup R_3^{red}$, any set agents $r_{i'} \in R^{set}$ must be in a room with fraction $\frac{2(5 \cdot j' + 1)}{s}$ where $i' \in A_{j'-3}$. There are exactly $s$ agents that approve of fraction $\frac{2(5 \cdot j' + 1)}{s}$, namely the agents in $\{r_i, \tilde{r}_i\}_{i \in A_{j'-3}} \cup R_{j'}^{red} \cup B_{j'}^{fill}$, any other agent disapproves of that fraction. Since $|D_\pi^-| = 1$, the room $\pi(r_{i'})$ must contain at least $s-1$ agents from $\{r_i, \tilde{r}_i\}_{i \in A_{j'-3}} \cup R_{j'}^{red} \cup B_{j'}^{fill}$.
		
		Let us assume that there exists an agent $a \in \{r_i, \tilde{r}_i\}_{i \in A_{j'-3}} \cup R_{j'}^{red} \cup B_{j'}^{fill}$ such that $a \notin \pi(r_{i'})$. Note that under this assumption, the agent $r^- \in D_\pi^-$ must be contained in room $\pi(r_{i'})$. Therefore we have that $a \in D_\pi^+$. Consider the following cases.
		\begin{itemize}
			\item $a \in \{r_i, \tilde{r}_i\}_{i \in A_{j'-3}}$.\\
			W.l.o.g. let us write $a = r_{i''}$. Since $\pi$ does not contain a room that only consists of red agents, we must have $\theta(\pi(a)) = \frac{2(5 \cdot j'' + 1)}{s}$ where $i'' \in A_{j''-3}$. There are exactly $s$ agents that approve of fraction $\frac{2(5 \cdot j'' + 1)}{s}$, namely the agents in $\{r_{i}, \tilde{r}_{i}\}_{i \in A_{j''-3}} \cup R_{j''}^{red} \cup B_{j''}^{fill}$, any other agent disapproves of that fraction.
			 
			Note that $\tilde{r}_{i''} \in \{r_{i}, \tilde{r}_{i}\}_{i \in A_{j''-3}}$ and $\tilde{r}_{i''} \in \pi(r_{i'})$. Since $a \notin \pi(r_{i'})$, we have that $\tilde{r}_{i''} \notin \pi(a)$. Thus the room $\pi(a)$ contains at most $s-1$ agents from $\{r_{i}, \tilde{r}_{i}\}_{i \in A_{j''-3}} \cup R_{j''}^{red} \cup B_{j''}^{fill}$. Therefore room $\pi(a)$ contains at least 1 agent that disapproves of fraction $\frac{2(5 \cdot j'' + 1)}{s}$.
			
			Since both room $\pi(r_{i'})$ and room $\pi(a)$ contain an agent that disapproves of the fraction of its assigned room, we have that $|D_\pi^-| \geq 2$. This contradicts $|D_\pi^-| = 1$, thus this case cannot occur.
			
			\item $a \in  R_{j'}^{red} \vee a \in  B_{j'}^{fill}$.\\
			We have that $\theta(\pi(a)) = \frac{2(5 \cdot j' + 1)}{s}$ or  $\theta(\pi(a)) = \frac{2(5 \cdot j' - 2)}{s}$. 
			
			There are exactly $s$ agents that approve of fraction $\frac{2(5 \cdot j' + 1)}{s}$, namely the agents in $\{r_i, \tilde{r}_i\}_{i \in A_{j'-3}} \cup R_{j'}^{red} \cup B_{j'}^{fill}$, any other agent disapproves of that fraction.

			There are exactly $s$ agents that approve of fraction $\frac{2(5 \cdot j' - 2)}{s}$, namely the agents in $B_{j'}^{add} \cup R_{j'}^{red} \cup B_{j'}^{fill}$, any other agent disapproves of that fraction.
			
			Since $R_{j'}^{red} \cup B_{j'}^{fill} \setminus \{a\} \subseteq \pi(r_{i'})$, the room $\pi(a)$ must contain at least $|R_{j'}^{red} \cup B_{j'}^{fill} \setminus \{a\}| = s-7 > 2$ agents that disapprove of the fraction $\theta(\pi(a))$. This contradicts $|D_\pi^-| = 1$, thus this case cannot occur.
			
		\end{itemize}
		
		Since there does not exist an agent $a \in \{r_i, \tilde{r}_i\}_{i \in A_{j'-3}} \cup R_{j'}^{red} \cup B_{j'}^{fill}$ such that $a \notin \pi(r_{i'})$, we have that
		$$\{r_i, \tilde{r}_i\}_{i \in A_{j'-3}} \cup R_{j'}^{red} \cup B_{j'}^{fill} \in \pi.$$
		
		Let $A$ denote the collection of 3-sets $A_{j-3} \in C$ such that 
		$\{r_i, \tilde{r}_i\}_{i \in A_{j-3}} \cup R_{j}^{red} \cup B_{j}^{fill} \in \pi,$
		i.e., 
		$$A = \{A_{j-3} \in C | \{r_i, \tilde{r}_i\}_{i \in A_{j-3}} \cup R_{j}^{red} \cup B_{j}^{fill} \in \pi\}.$$ In a valid outcome $\pi$, every set agent is assigned to exactly 1 room. Therefore no two 3-sets in $A$ share an element, i.e., for any $A_i, A_j \in A$ such that $i \neq j$, we have that $A_i\cap A_j = \emptyset$. 
		
		From the above, we have that $A$ must be a solution of $(X,C)$. Therefore we have some solution $A = C' \subseteq C$ of $(X,C)$ where for each $A_{j-3} \in C'$ it holds that
		$$\{r_i, \tilde{r}_i\}_{i \in A_{j-3}} \cup R_j^{red} \cup B_j^{fill} \in \pi.$$
		
		\item\label{lem25prop3} $B_j^{add} \cup R_j^{red} \cup B_j^{fill} \in \pi$ for each $A_{j-3} \notin C'$.\\
		Since $D_\pi^n \cup D_\pi^- \subseteq \hat{R}^{set} \cup R_3^{red}$, any redundant agent $r_{j'}^p \in R_{j'}^{red}$, where $A_{j'-3}\notin C'$, must be in a room with fraction $\frac{2(5 \cdot j' + 1)}{s}$ or $\frac{2(5 \cdot j' - 2)}{s}$. 
		
		There are exactly $s$ agents that approve of fraction $\frac{2(5 \cdot j' + 1)}{s}$, namely the agents in $\{r_i, \tilde{r}_i\}_{i \in A_{j'-3}} \cup R_{j'}^{red} \cup B_{j'}^{fill}$, any other agent disapproves of that fraction.

		There are exactly $s$ agents that approve of fraction $\frac{2(5 \cdot j' - 2)}{s}$, namely the agents in $B_{j'}^{add} \cup R_{j'}^{red} \cup B_{j'}^{fill}$, any other agent disapproves of that fraction.
		
		By property 2, all the set agents $r_{i'} \in R^{set}, \tilde{r}_{i'} \in \tilde{R}^{set}$ must be in the room $\{r_i, \tilde{r}_i\}_{i \in A_{j''-3}} \cup R_{j''}^{red} \cup B_{j''}^{fill} \in \pi$ where $i' \in A_{j''-3}$ and $A_{j''-3} \in C'$. Thus there are exactly $2(5j'-2)$ remaining red agents that approve of fraction $\frac{2(5 \cdot j' + 1)}{s}$. Therefore room $\pi(r_{j'}^p)$ cannot have fraction $\frac{2(5 \cdot j' + 1)}{s}$, as otherwise we would have at least 3 agents in $D_\pi^-$. Therefore room $\pi(r_{j'}^p)$ must be of fraction $\frac{2(5 \cdot j' - 2)}{s}$.
		
		Since $|D_\pi^-| = 1$, the room $\pi(r_{j'}^p)$ contains at least $s-1$ agents from $R_{j'}^{red} \cup B_{j'}^{fill} \cup B_{j'}^{add}$. Assume that there exists an agent $a \in R_{j'}^{red} \cup B_{j'}^{fill} \cup B_{j'}^{add}$ such that $a \notin \pi(r_{j'}^p)$. Note that under this assumption, the agent $r^- \in D_\pi^-$ must be contained in room $\pi(r_{j'}^p)$. Therefore we have that $a \in D_\pi^+$. Additionally, since $D_\pi^n \cup D_\pi^- \subseteq \hat{R}^{set} \cup R_3^{red}$, we have that $a$ is a red agent. Thus we have that $a\in R_{j'}^{red}$. Since $a\in R_{j'}^{red}$ and $a \in D_\pi^+$, we have that $\theta(\pi(a)) = \frac{2(5 \cdot j' + 1)}{s}$ or $\theta(\pi(a)) = \frac{2(5 \cdot j' - 2)}{s}$.
		
		There are exactly $s$ agents that approve of fraction $\frac{2(5 \cdot j' + 1)}{s}$, namely the agents in $\{r_i, \tilde{r}_i\}_{i \in A_{j'-3}} \cup R_{j'}^{red} \cup B_{j'}^{fill}$, any other agent disapproves of that fraction.

		There are exactly $s$ agents that approve of fraction $\frac{2(5 \cdot j' - 2)}{s}$, namely the agents in $B_{j'}^{add} \cup R_{j'}^{red} \cup B_{j'}^{fill}$, any other agent disapproves of that fraction.
		
		Since $R_{j'}^{red} \cup B_{j'}^{fill} \setminus \{a\} \subseteq \pi(r_{j'}^p)$, the room $\pi(a)$ must contain at least $|R_{j'}^{red} \cup B_{j'}^{fill} \setminus \{a\}| = s-7 > 2$ agents that disapprove of the fraction $\theta(\pi(a))$. This contradicts $|D_\pi^-| = 1$, thus there cannot exist an agent $a \in R_{j'}^{red} \cup B_{j'}^{fill} \cup B_{j'}^{add}$ such that $a \notin \pi(r_{j'}^p)$. Therefore we have that
		$$R_{j'}^{red} \cup B_{j'}^{fill} \cup B_{j'}^{add} \in \pi.$$
		
		The above holds for any redundant agent $r_j^p \in R_j^{red}$ where $A_{j-3} \notin C'$. Thus we have that for each $A_{j-3} \notin C'$
		$$B_j^{add} \cup R_j^{red} \cup B_j^{fill} \in \pi.$$
		
		\item $\{r^n\} \cup R_2^{red} \cup B_2^{fill} \in \pi$.\\
		Let us define the room $S^n = \pi(r^n)$. We have that $\theta(S^n) = \frac{5\cdot 2 - 1}{s}$ as it is the only fraction contained in any $D_a^n$. Since $r^- \in D_\pi^-$ and $r^- \in \hat{R}^{set} \cup R_3^{red}$, we have that $r^- \notin S^n$. If $r^- \in S^n$, we would have $r^- \in D_\pi^n$. Thus the remaining agents in $S^n$ must be contained in $D_\pi^+$. 
		
		The red agents that approve of fraction $\frac{5\cdot 2 - 1}{s}$ are exactly the agents in $R_2^{red}$ of which there are $5\cdot 2 -2$, any other red agent does not approve of that fraction. Thus the red agents in $S^n$ are exactly $r^n$ and $R_2^{red}$. 
		
		The blue agents that approve of fraction $\frac{5\cdot 2 - 1}{s}$ are exactly $B_2^{fill} \cup B_2^{add}$ of which there are $s-(5\cdot 2 - 2)$, any other blue agent disapproves of that fraction.
		
		Assume that $\{\tilde{b}_2\} = B_2^{add} \subseteq S^n$. This means that there exists a blue agent $b \in B_2^{fill}$ such that $b \notin S^n$. Since $D_\pi^n \cup D_\pi^- \subseteq \hat{R}^{set} \cup R_3^{red}$, we have that $b \in D_\pi^+$. Therefore we have $\theta(\pi(b)) = \frac{5\cdot 2 - 1}{s}$ or $\theta(\pi(b)) = \frac{5\cdot 2 - 2}{s}$.
		
		There are exactly $s-1$ agents that approve of fraction $\frac{5 \cdot 2 - 1}{s}$, namely the agents in $R_{2}^{red} \cup B_{2}^{fill}$, any other agent does not approve of that fraction.
		
		There are exactly $s$ agents that approve of fraction $\frac{5 \cdot 2 - 2}{s}$, namely the agents in $B_{2}^{add} \cup R_{2}^{red} \cup B_{2}^{fill}$, any other agent disapproves of that fraction.
		
		Since $R_{2}^{red} \cup B_{2}^{fill} \setminus \{b\} \subseteq S^n$, the room $\pi(b)$ must contain at least $|R_{2}^{red} \cup B_{2}^{fill} \setminus \{b\}| = s-2 >2$ agents that do not approve of the fraction $\theta(\pi(b))$. This contradicts $|D_\pi^-|=1$ and $|D_\pi^n|=1$, thus $\{\tilde{b}_2\} = B_a^{add} \nsubseteq S^n$.
		
		Therefore the blue agents in room $S^n$ must be exactly the agents in $B_2^{fill}$ and we have
		$$\{a^n\} \cup R_2^{red} \cup B_2^{fill} \in \pi.$$
		
		\item $(\hat{R}^{set} \cup R_3^{red}) \setminus \{r^n, r^-\} \cup B_3^{fill} \in \pi$.\\
		Let $\hat{r}_j \in \hat{R}^{set}$ where $1 \leq j \leq 3$ such that $\hat{r}_j \in D_\pi^+$ and let $S^j$ denote the room in $\pi$ that contains $\hat{r}_j$. Since $\hat{r}_j \in D_\pi^+$, we have that $\theta(S^j) = \frac{5\cdot 3 - 1}{s}$. Note that $S^j$ cannot contain the agents $r^n$ and $r^-$ as they approve of fraction $\frac{5\cdot 3 - 1}{s}$. Thus every agent in $S^j$ must approve of fraction $\frac{5\cdot 3 - 1}{s}$.
		
		The red agents that approve of fraction $\frac{5\cdot 3 - 1}{s}$ that may be contained in $S^j$ are exactly $(\hat{R}^{set} \cup R_3^{red}) \setminus \{r^n, r^-\}$. There are $5 \cdot 3 -1$ agents in $(\hat{R}^{set} \cup R_3^{red}) \setminus \{r^n, r^-\}$. Thus the red agents in $S^j$ must be exactly $(\hat{R}^{set} \cup R_3^{red}) \setminus \{r^n, r^-\}$. 
		
		The blue agents that approve of fraction $\frac{5\cdot 3 - 1}{s}$ are exactly $B_3^{fill}\cup B_3^{add}$ of which there are $s-(5\cdot 3-2)$.
		
		Assume that $\{\tilde{b}_3\} = B_3^{add} \subseteq S^j$. This means that there exists a blue agent $b \in B_3^{fill}$ such that $b \notin S^j$. Since $D_\pi^n \cup D_\pi^- \subseteq \hat{R}^{set} \cup R_3^{red}$, we have that $b \in D_\pi^+$. Therefore we have $\theta(\pi(b)) = \frac{5\cdot 3 - 1}{s}$ or $\theta(\pi(b)) = \frac{5\cdot 3 - 2}{s}$.
		
		There are exactly $s+2$ agents that approve of fraction $\frac{5 \cdot 3 - 1}{s}$, namely the agents in $\hat{R}^{set} \cup R_{3}^{red} \cup B_{3}^{fill}$, any other agent disapproves of that fraction.
		
		There are exactly $s$ agents that approve of fraction $\frac{5 \cdot 3 - 2}{s}$, namely the agents in $B_{3}^{add} \cup R_{3}^{red} \cup B_{3}^{fill}$, any other agent disapproves of that fraction.
		
		Since $(\hat{R}^{set} \cup R_3^{red}) \setminus \{r^n, r^-\} \cup B_{3}^{fill} \setminus \{b\} \subseteq S^j$, the room $\pi(b)$ must have at least $|(\hat{R}^{set} \cup R_3^{red}) \setminus \{r^n, r^-\} \cup B_{3}^{fill} \setminus \{b\}| - 2 = s - 5 > 2$ agents that disapprove of the fraction $\theta(\pi(b))$. This contradicts $|D_\pi^-| = 1$, thus $\{\tilde{b}_3\} = B_3^{add} \nsubseteq S^j$.
		
		Therefore the blue agents in room $S^n$ must be exactly $B_3^{fill}$ and we have
		$$(\hat{R}^{set} \cup R_3^{red}) \setminus \{a^n, a^-\} \cup B_3^{fill} \in \pi.$$
		
		\item\label{lem25prop6} $B^{even} \cup \bigcup\limits_{j \in [3]}B_j^{add} \cup \bigcup\limits_{A_{j-3}\in C'}B_j^{add} \in \pi$.\\
		Since $D_\pi^n \cup D_\pi^- \subseteq \hat{R}^{set} \cup R_3^{red}$, we have that $B^{even}\subseteq D_\pi^+$. Thus any evening agent $b \in B^{even}$ must be in a room of fraction $0$.
		
		There are exactly $s + 6q + 3$ agents that approve of fraction $0$, namely the agents in $B^{even} \cup B^{mon} \cup B^{add}$, any other agent disapproves of that fraction. Since $s < s + 6q + 3 < 2s$, we have at most 1 room with fraction $0$. Otherwise we have at least $2s - (s + 6q + 3) = 4q + 42 + 2m$ agents in $D_\pi^-$ which would contradict $|D_\pi^-| = 1$. Thus every evening agent must be assigned to the same room, i.e., for any $b', b'' \in B^{even}$ we have $\pi(b') = \pi(b'')$.
		
		By Property \ref{lem25prop1} the room of evening agent $b$ cannot contain agents from $B^{mon}$. By Property \ref{lem25prop3} the room of evening agent $b$ cannot contain agents from $\bigcup\limits_{A_j \notin C'}B_j^{add}$. Since $D_\pi^n \cup D_\pi^- \subseteq \hat{R}^{set} \cup R_3^{red}$, no blue agent may be assigned to a room by $\pi$ with a fraction that it disapproves of. Therefore we have $\pi(b) \subseteq B^{even} \cup \bigcup\limits_{j \in [3]}B_j^{add} \cup \bigcup\limits_{A_{j-3}\in C'}B_j^{add}$. Since there are $s-2m-3$ evening agents, $\pi(b)$ must contain exactly $2m+3$ agents from $\bigcup\limits_{j \in [3]}B_j^{add} \cup \bigcup\limits_{A_{j-3}\in C'}B_j^{add}$. Note that $|\bigcup\limits_{j \in [3]}B_j^{add}| = 3$ and $|\bigcup\limits_{A_j \in C'}B_j^{add}|= \frac{m}{3}\cdot 6 = 2m$. Thus we have that
		$$B^{even} \cup \bigcup\limits_{j \in [3]}B_j^{add} \cup \bigcup\limits_{A_{j-3}\in C'}B_j^{add} \in \pi.$$
		
		\item\label{lem25prop7} $\{r^-\} \cup R_1^{red} \cup B_1^{fill} \in \pi$.\\
		By Properties \ref{lem25prop1}-\ref{lem25prop6}, we have exactly $s$ remaining agents of which their assigned room has not been specified yet. These agents are $\{r^-\} \cup R_1^{red} \cup B_1^{fill}$, where $r^- \in \hat{R}^{set} \cup R_3^{red}$. Thus we have that $$\{r^-\} \cup R_1^{red} \cup B_1^{fill} \in \pi.$$
\end{enumerate}

By Properties \ref{lem25prop1}-\ref{lem25prop7}, we have that $\pi$ must be a reduced-type outcome. \qedhere
\end{proof}

\begin{lemma}\label{lem3.9}
	For any non-reduced-type outcome $\pi$, there exists some reduced-type outcome $\pi'$ that is strictly more popular than $\pi$.
\end{lemma}
\begin{proof}
	Let $\pi$ be an arbitrary non-reduced-type outcome.
	By \cref{lemma3.7} we have that $|D_{\pi}^n| \neq 1 \vee |D_{\pi}^-| \neq 1$. Consider the following cases.
	\begin{enumerate}
		\item $|D_{\pi}^-| \neq 1$.\\\
		By \cref{lemma3.2} and case assumption, we have that $|D_{\pi}^-| \geq 2$. Consider the following cases.
		\begin{enumerate}[\theenumi.1.]
			\item $|D_{\pi}^-| = 2$.\\
			Consider the following cases.
			\begin{enumerate}[\theenumi.\theenumii.1.]
				\item $|D_{\pi}^n| = 0$.\\
				By \cref{lemma3.6}, we have that $D_{\pi}^- \subseteq \hat{R}^{set} \cup R_3^{red}$. Let us write $D_\pi^- = \{a_1, a_2\}$. Since $a_1, a_2 \in \hat{R}^{set} \cup R_3^{red}$, there exists a reduced-type outcome $\pi'$ such that $D_{\pi'}^- = \{a_1\}$ and $D_{\pi'}^n = \{a_2\}$. Thus we have that $\phi(\pi', \pi) = 1$.
				
				\item $|D_{\pi}^n| \geq 1$.\\
				Let us write $D_\pi^- = \{a_1, a_2\}$ and $a_3 \in D_\pi^n$. We have that $a_3 \in \hat{R}^{set} \cup R_3^{red}$, since the agents in $\hat{R}^{set} \cup R_3^{red}$ are the only agents with a corresponding non-empty $D_a^n$. 
				
				Let $\pi'$ be a reduced-type outcome such that $\{a_3\} = D_{\pi'}^n$ and $\{a^-\} = D_{\pi'}^-$. Note that $N(\pi, \pi') \subseteq \{a^-\}$. Consider the following cases.
				\begin{enumerate}[\theenumi.\theenumii.\theenumiii.1.]
					\item $N(\pi, \pi') =  \emptyset$.\\
					Then we have that $a^- \in D_\pi^-$. W.l.o.g. assume that $a_1 = a^-$. We have that $a_2 \in D_{\pi'}^+$, thus $a_2 \in N(\pi', \pi)$. Thus we have that $\phi(\pi',\pi) > 0$, i.e., $\pi'$ is more popular than $\pi$.
					
					\item $N(\pi, \pi') = \{a^-\}$.\\
					Then we have that $a^- \in D_\pi^n \cup D_\pi^+$ and $a_1,a_2 \in D_{\pi'}^+$. Thus $a_1,a_2 \in N(\pi',\pi)$. Therefore we have that $\phi(\pi',\pi) > 0$, i.e., $\pi'$ is more popular than $\pi$.
				\end{enumerate}
			\end{enumerate}
			
			\item $|D_{\pi}^-| > 2$.\\
			Let us write $a_1, a_2, a_3 \in D_\pi^-$. Let $\pi'$ be an arbitrary reduced-type outcome with $\{a^n\} = D_{\pi'}^n$ and $\{a^-\} = D_{\pi'}^-$. Note that $N(\pi, \pi') \subseteq \{a^n, a^-\}$. Consider the following cases.
			\begin{enumerate}[\theenumi.\theenumii.1.]
				\item $a^- \in  D_\pi^-$.\\
				W.l.o.g. assume that $a_1 = a^-$. In this case, we have that $N(\pi, \pi') \subseteq \{a^n\}$. We have that $a_2, a_3 \in N(\pi', \pi)$. Thus $\phi(\pi', \pi) > 0$. 
				
				\item $a^- \notin  D_\pi^-$.\\
				In this case, we have that $N(\pi, \pi') \subseteq \{a^-, a^n\}$. We have that $a_1, a_2, a_3 \in N(\pi', \pi)$. Thus $\phi(\pi', \pi) > 0$. 
				
			\end{enumerate}
		\end{enumerate}
		
		\item $|D_{\pi}^n| \neq 1$.\\\
		Consider the following cases.
		\begin{enumerate}[\theenumi.1.]
			\item $|D_{\pi}^n| = 0$.\\
			By \cref{lemma3.1}, we have that $|D_{\pi}^-| \geq 2$. This case is proven similarly to case 1.1.1. as we have the exact same case assumption, i.e., $|D_{\pi}^n| = 0$ and $|D_{\pi}^-| \geq 2$.
			
			\item $|D_{\pi}^n| \geq 2$.\\
			Let us write $a_1, a_2 \in D_\pi^n$. We have that $\{a_1, a_2\} \subseteq \hat{R}^{set} \cup R_3^{red}$, since the agents in $\hat{R}^{set} \cup R_3^{red}$ are the only agents with a corresponding non-empty $D_a^n$.
			By \cref{lemma3.2}, we have that $|D_{\pi}^-| \geq 1$. 
			Let us write $a_3 \in D_\pi^-$. Let us consider the following cases regarding $a_3$.
			\begin{enumerate}[\theenumi.\theenumii.1.]
				\item $a_3 \in \hat{R}^{set} \cup R_3^{red}$.\\
				By the definition of reduced-type outcome, we can construct a reduced-type outcome $\pi'$ such that $D_{\pi'}^- = \{a_3\}$, $D_{\pi'}^n = \{a_2\}$, and $a_1 \in D_{\pi'}^+$. 
				
				We have that $N(\pi, \pi') = \emptyset$ and $|N(\pi',\pi)| \geq 1$. Therefore $\phi(\pi', \pi) > 0$.
				
				\item $a_3 \notin \hat{R}^{set} \cup R_3^{red}$.\\
				By the definition of reduced-type outcome, we can construct a reduced-type outcome $\pi'$ such that $D_{\pi'}^- = \{a^-\}$, $D_{\pi'}^n = \{a_1\}$, and $a_2, a_3 \in D_{\pi'}^+$. 
				
				We have that $N(\pi,\pi') \subseteq \{a^-\}$. Additionally we have that $a_2, a_3 \in N(\pi',\pi)$, therefore $\phi(\pi', \pi) > 0$. \qedhere
			\end{enumerate}
		\end{enumerate}
	\end{enumerate}
\end{proof}

\begin{lemma}\label{lemma4}
	An reduced-type outcome $\pi$ is not popular.
\end{lemma}
\begin{proof}
	Let $\pi$ be some reduced-type outcome. Since $\pi$ is a reduced-type outcome, we can write 
	$$\{a_j\} \cup R^{red}_j \cup B_j^{fill} \in \pi \text{ for $j \in [2]$,}$$ 
	and 
	$$\{a_3, a_4, a_5\} \cup (R^{red}_3 \setminus \{a_1, \dots, a_5\}) \cup B_3^{fill} \in \pi,$$
	where $\{a_1, \dots, a_5\} \subseteq \hat{R}^{set} \cup R_3^{red}$.
	
	Let us construct an outcome $\pi'$ as follows:
	\begin{align*}
		& 	& \pi'= \pi &\setminus \{\{a_1\} \cup R^{red}_1 \cup B_1^{fill}\} \\
		& 	& 		& \setminus \{\{a_2\} \cup R^{red}_2 \cup B_2^{fill}\} \\
		& 	& 		& \setminus \{\{a_3, a_4, a_5\} \cup (R^{red}_3 \setminus \{a_1, \dots, a_5\}) \cup B_3^{fill}\} \\
		\\
		&	&		& \cup \{\{a_3\} \cup R^{red}_1 \cup B_1^{fill}\} \\
		& 	& 		& \cup \{\{a_1\} \cup R^{red}_2 \cup B_2^{fill}\} \\
		& 	& 		& \cup \{\{a_2, a_4, a_5\} \cup (R^{red}_3 \setminus \{a_1, \dots, a_5\}) \cup B_3^{fill}\}.
	\end{align*}
	That is, we ``rotate'' the agents $a_1, a_2, a_3$ in $\pi$ to obtain $\pi'$. We have that $N(\pi', \pi) = \{a_1, a_2\}$ and $N(\pi, \pi') = \{a_3\}$. Thus $\pi'$ is more popular than $\pi$. Note that $\pi'$ itself is also an reduced-type outcome.
\end{proof}

\begin{theorem}
	Determining whether a popular outcome exists in a roommate diversity game is co-NP-hard, even if the preferences are trichotomous.	
\end{theorem}
\begin{proof}
	From \cref{lemma0}, if $(X,C)$ has no solution, then we have a popular outcome, namely $\pi_{mon}$.
	From \cref{lem3.9} and \ref{lemma4}, if $(X,C)$ has a solution, then we have no popular outcome. Therefore determining the existence of a popular outcome for a Roommate diversity Game is co-NP-hard.
\end{proof}

%% file: sections/conclusion.tex
\section{Conclusion}
We have demonstrated that determining the existence of a (strictly) popular outcome for a roommate diversity game is co-NP-hard and computing a mixed popular outcome is not possible in polynomial time, unless P=NP. Even when the preferences are tri- or dichotomous, the problem remains intractable. As we have only demonstrated hardness, a potential avenue for future research would be to demonstrate completeness for a certain complexity class. We conjecture that the problem is $\Pi_2^p$-complete.

A popular outcome is guaranteed to exist in a roommate diversity game, when the room size is fixed to 2. Additionally, it is possible to compute a popular outcome in polynomial time. As the problem becomes tractable when fixing the room size to 2, it may be possible to construct a fixed-parameter tractable algorithm with the room size as parameter that determines the existence of a (strictly) popular outcome in a roommate diversity game. This is also left for future research.

%% file: lipics-v2021-sample-article.bbl
\begin{thebibliography}{1}

\bibitem{10.5555/3491440.3491454}
Niclas Boehmer and Edith Elkind.
\newblock Stable roommate problem with diversity preferences.
\newblock In {\em Proceedings of the Twenty-Ninth International Joint
  Conference on Artificial Intelligence}, IJCAI'20, 2021.

\bibitem{10.1613/jair.1.13470}
Felix Brandt and Martin Bullinger.
\newblock Finding and recognizing popular coalition structures.
\newblock {\em J. Artif. Int. Res.}, 74, sep 2022.
\newblock \href {https://doi.org/10.1613/jair.1.13470}
  {\path{doi:10.1613/jair.1.13470}}.

\bibitem{10.1145/3155301}
Ran Duan, Seth Pettie, and Hsin-Hao Su.
\newblock Scaling algorithms for weighted matching in general graphs.
\newblock {\em ACM Trans. Algorithms}, 14(1), jan 2018.
\newblock \href {https://doi.org/10.1145/3155301} {\path{doi:10.1145/3155301}}.

\bibitem{10.5555/574848}
Michael~R. Garey and David~S. Johnson.
\newblock {\em Computers and Intractability; A Guide to the Theory of
  NP-Completeness}.
\newblock W. H. Freeman \& Co., USA, 1990.

\bibitem{IRVING1985577}
Robert~W Irving.
\newblock An efficient algorithm for the “stable roommates” problem.
\newblock {\em Journal of Algorithms}, 6(4):577--595, 1985.
\newblock URL:
  \url{https://www.sciencedirect.com/science/article/pii/0196677485900331},
  \href {https://doi.org/https://doi.org/10.1016/0196-6774(85)90033-1}
  {\path{doi:https://doi.org/10.1016/0196-6774(85)90033-1}}.

\bibitem{doi:10.1137/0404023}
Cheng Ng and Daniel~S. Hirschberg.
\newblock Three-dimensional stable matching problems.
\newblock {\em SIAM Journal on Discrete Mathematics}, 4(2):245--252, 1991.
\newblock \href {http://arxiv.org/abs/https://doi.org/10.1137/0404023}
  {\path{arXiv:https://doi.org/10.1137/0404023}}, \href
  {https://doi.org/10.1137/0404023} {\path{doi:10.1137/0404023}}.

\bibitem{RONN1990285}
Eytan Ronn.
\newblock Np-complete stable matching problems.
\newblock {\em Journal of Algorithms}, 11(2):285--304, 1990.
\newblock URL:
  \url{https://www.sciencedirect.com/science/article/pii/0196677490900072},
  \href {https://doi.org/https://doi.org/10.1016/0196-6774(90)90007-2}
  {\path{doi:https://doi.org/10.1016/0196-6774(90)90007-2}}.

\bibitem{v.Neumann1928}
J.~v.~Neumann.
\newblock Zur theorie der gesellschaftsspiele.
\newblock {\em Mathematische Annalen}, 100(1):295--320, Dec 1928.
\newblock \href {https://doi.org/10.1007/BF01448847}
  {\path{doi:10.1007/BF01448847}}.

\end{thebibliography}
